\documentclass[12pt]{amsart}
\usepackage{amsfonts}
\usepackage{graphicx}
\usepackage{tabularx}
\usepackage{array}
\usepackage[usenames,dvipsnames]{color}
\usepackage{comment}
\usepackage{amsmath}
\usepackage{amsthm}
\usepackage{amssymb}
\usepackage{fullpage}
\usepackage[dvipsnames]{xcolor}
\usepackage{tikz}
\usepackage{listings}

\DeclareMathOperator{\conv}{conv}

\DeclareMathOperator{\dist}{dist}
\DeclareMathOperator{\diam}{diam}

\DeclareMathOperator{\R}{\mathbb R}

\newtheorem{theorem}{Theorem}[section]
\newtheorem{proposition}[theorem]{Proposition}

\newtheorem{conjecture}[theorem]{Conjecture}
\newtheorem{claim}[theorem]{Claim}

\newtheorem{lemma}[theorem]{Lemma}
\newtheorem{corollary}[theorem]{Corollary}
\theoremstyle{definition}
\newtheorem{definition}[theorem]{Definition}

\newtheorem{question}[theorem]{Question}

\newcommand{\mb}{\mathbf}
\newcommand{\dotp}{\boldsymbol{\cdot}}

\def\epsilon{\varepsilon}

\title{A cornering strategy for synchronizing a DFA}

\author{Peter Bradshaw}
\address{Department of Mathematics, University of Illinois Urbana-Champaign}
\email{pb38@illinois.edu}
\thanks{We thank Tom Shermer and Jan Manuch for insightful initial discussions on this problem, particularly as they pertain to product DFAs and difference DFAs in $\mathbb{R}^2$. This project was partially funded by NSF RTG grant DMS-1937241, and by NSERC through Discovery Grant R611368.}

\author{Alexander Clow}
\address{Department of Mathematics, Simon Fraser University}
\email{alexander\_clow@sfu.ca}

\author{Ladislav Stacho}
\address{Department of Mathematics, Simon Fraser University}
\email{ladislav\_stacho@sfu.ca}

\date{\today}

\begin{document}

\begin{abstract}
    This paper considers the existence of short synchronizing words in deterministic finite automata (DFAs).
    We define two general strategies for generating synchronizing words, and we show that each of these strategies can be applied if and only if a DFA is synchronizable.
    Furthermore, we show that if a synchronizable DFA is well-structured, then our strategies generate short synchronizing words.
    The first of our strategies,
    called the \emph{cornering strategy}, 
    takes advantage of states in a DFA with properties similar to those of a polytope vertex.
    The second of our strategies, similar to the cornering strategy and called the \emph{$f$-ordered strategy},
    takes advantage of a partial order defined on the states of a DFA.
    
    We apply our cornering strategy
    to the class of \emph{difference DFAs},
    whose states form subsets of $\mathbb R^d$
    and whose input symbols correspond to translation vectors between states.
    We show that difference DFAs share many similarities with aperiodic DFAs, and in particular, a difference DFA $M$ has a synchronizing word if and only if it has a universally reachable state.
    Using the cornering strategy, we also show that 
    under certain conditions,
    such an $n$-state DFA $M$
    has a synchronizing word of length at most $(n-1)^2$ and thereby satisfies \v{C}ern\'y's conjecture. 
    Using the $f$-ordered strategy,
    we also show that a synchronizable DFA
    whose states have a certain partial order that is preserved by a set of short words
    also has a short synchronizing word, and we consider several consequences of this result.
    
    Finally, we consider how the cornering strategy can be applied to the problem of synchronizing the product of two DFAs $M_1, M_2$ that share a common alphabet, and we show that the product $M_1 \times M_2$
    often has a synchronizing word that is subquadratic in the number of states of $M_1 \times M_2$.
\end{abstract}
\maketitle

\section{Introduction}

A \textit{deterministic finite automaton} (or DFA) is a quintuple $ M = (Q,\Sigma,\delta,s, A)$ such that $Q$ is a finite \emph{state set}, $\Sigma$ is an \emph{input alphabet}, $\delta: Q \times \Sigma \rightarrow Q$ is a \emph{transition function}, $s \in Q$ is a \emph{start state}, and $A \subseteq Q$ is a set of \emph{accepting states} (or \emph{final states}).
Elements of $Q$ and $\Sigma$ are called \emph{states} and \emph{input symbols}, respectively.
The DFA $M$ also has an \emph{active state},
which is an element $q \in Q$ that identifies
the current state of the system.
Initially, the active state of $M$ is $s$;
then, given an active state $q$ of $M$, the active state 
can be updated to a new active state $q(\tau)$ by the action of an input word $\tau =  c_1\cdots c_k$ on $q$, where $ c_1,\dots,  c_k \in \Sigma$. Here, 
$$
q( c_1) = \delta(q, c_1)
$$
and 
$$
q(\tau) = \delta(q( c_1\cdots c_{k-1}),  c_k).
$$
If $\tau$ is the \emph{empty word} $\epsilon$ (the word of length $0$), then
$q(\tau) = q (\epsilon) = q$.
We say that $|Q|$ is the \emph{order} of $M$.
Note that every DFA $M$ that we consider is \emph{complete}, so that $\delta(q,c)$ is well defined for each $q \in Q$ and $c \in \Sigma$.
We sometimes
omit the starting state $s$ and the accepting state set $A$ when they are not necessary for our purposes; in this case, we write $M = (Q,\Sigma,\delta)$.

We write $\Sigma^*$ for the set of words generated by the alphabet $\Sigma$.
Note that each word $\tau \in \Sigma^*$ corresponds to a function $Q \rightarrow Q$ given by the action of $\tau$ on $Q$---that is, the function that maps each $q \in Q$ to $q(\tau)$.
These functions form a semigroup, and the empty word $\epsilon$ is an identity element of this semigroup.
This semigroup is called the \emph{transition monoid} of $M$.
As $Q$ is finite, it follows that the transition monoid of $M$ has at most $|Q|^{|Q|}$
distinct
elements and hence is finite.
On the other hand, given a finite monoid $S$, 
one may construct a DFA $M = (Q,\Sigma,\delta)$ with
transition monoid $S$ by letting $Q = \Sigma = S$ and 
letting $\delta(q, s) = qs$ for each $q \in Q$ and $s \in S$.
Therefore, DFAs and finite monoids are closely related.

Given a word $\tau \in \Sigma^*$, we say that the DFA $M = (Q,\Sigma,\delta,s,A)$ \emph{accepts} $\tau$ if $s(\tau) \in A$; otherwise, we say that $M$ \emph{rejects} $\tau$.
We write $L(M)$ for the set of all words $\tau \in \Sigma^*$ accepted by $M$, and we say that $L(M)$ is the \emph{language accepted by $M$}.
The DFA $M$ is \emph{minimal} if all states of $Q$ are reachable (for each $q \in Q$, there exists $\tau \in \Sigma^*$ satisfying $s(\tau) = q$)
and each pair of states of $Q$ is distinguishable (for each distinct pair $p,q \in Q$, there exists $\tau \in \Sigma^*$ for which exactly one of $p(\tau), q(\tau)$ is in $A$).
It is well known (see Theorem \ref{thm:unique-min}) that for each regular language $L \subseteq \Sigma^*$,
there is a unique minimal DFA $M$ for which $L(M) = L$. 
Therefore, for each alphabet $\Sigma$, there is a bijection between regular 
languages of $\Sigma^*$ and minimal DFAs with alphabet $\Sigma$.

A DFA $M = (Q,\Sigma,\delta)$ can be represented as an arc-labeled directed multigraph $G$ as follows.
We let $V(G) = Q$, so that the vertex set of $G$ is the state set of $M$.
Then, for each $q \in Q$ and $c \in \Sigma$,
we give $G$ an arc $(q, \delta(q,c))$ with the label $\psi((q,\delta(q,c))) = c$.
The multigraph $G$ is of regular out-degree $|\Sigma|$, and each arc $e = (p,q)$ with label $\psi(e) = c$ corresponds to the 
transition function assignment
$\delta(p,c ) = q$.
We also designate
a vertex $s \in V(G)$
as a \emph{start vertex} of $G$,
as well as a subset $A \subseteq V(G)$ of \emph{accepting vertices}.
As with the DFA $M$, we often omit $s$ and $A$ from the definition of $G$ when they are not necessary.
We often refer to $G$ simply as a \emph{digraph} with the understanding that $G$ may contain loops and parallel edges, and that the arcs of $G$ are labeled by the function $\psi$.

In a DFA represented by a digraph $G = (V,E,\psi)$, an active state corresponds to a vertex $q \in V(G)$.
We may define new active state $q(\tau)$ for an input word $\tau = c_1\cdots c_k \in \Sigma^*$ as the unique vertex $p \in V(G)$ such that there exists a directed walk $W = q, q_1,\dots,q_{k-1},p$ 
in $G$
satisfying 
$$
\psi(q,q_1)=c_1,\dots, \psi(q_i,q_{i+1})=c_{i+1},\dots, \psi(q_{k-1},p)=c_k.
$$
For example, given the DFA whose digraph is shown in Figure~\ref{Fig:DFA example}, if
the active state is $q$ and
the input is the word $\tau = \texttt{cc}$, then 
the active state updates to
$q(\tau)=p$. 
Given a walk $W$ with an edge sequence $(e_1, \dots, e_t)$, we often write $\psi(W) = \psi(e_1) \cdots \psi(e_t)$.
One can represent a DFA
with active state $q$ as a robot standing on the vertex $q$ in the associated digraph,
so that inputting the word $\tau$ to the robot commands the robot at $q$ to
follow the walk $W$ beginning at $q$ with $\psi(W)=\tau$.

\begin{figure}\label{Fig:DFA example}
\begin{center}
    \scalebox{0.85}{
        \begin{tikzpicture}[node distance={15mm}, thick, main/.style = {draw, circle}] 

\node[main][fill= black] (0) at (0,0) {}; 
    \node[fill=none] at (-1,0) (nodes) {$p$};
\node[main][fill= black] (1) at (0,2) {}; 
\node[main][fill= black] (2) at (0,4) {}; 
    \node[fill=none] at (-1,4) (nodes) {$q$};

\node[main][fill= black] (3) at (1.5,2) {}; 
\node[main][fill= black] (4) at (2.5,2) {}; 

\node[main][fill= black] (5) at (4,0) {}; 
\node[main][fill= black] (6) at (4,2) {}; 
\node[main][fill= black] (7) at (4,4) {};

\node[main][fill= white,draw=white] (z1) at (6,3) {};
\node[main][fill= white,draw=white] (z2) at (6,1) {};
\node[main][fill= white,draw=white] (y1) at (8,3) {};
\node[main][fill= white,draw=white] (y2) at (8,1) {};
\node[main][fill= white,draw=white] (x1) at (10,3) {};
\node[main][fill= white,draw=white] (x2) at (10,1) {};

\node[main][fill= white,draw=white] (z) at (6.5,2) {\texttt{a}};
\node[main][fill= white,draw=white] (z) at (8.5,2) {\texttt{b}};
\node[main][fill= white,draw=white] (z) at (10.5,2) {\texttt{c}};

\draw [->,blue] (z1) -- (z2);
\draw [->,ultra thick, dashed, red] (y1) -- (y2);
\draw [->,green!50!black, dotted] (x1) -- (x2);

\draw [->,blue] (2) -- (4);
\draw [->,ultra thick, dashed, red] (4) -- (3);
\draw [->,blue] (4) -- (0);

\draw [->,ultra thick, dashed, red] (2) -- (3);
\draw [->,ultra thick, dashed, red] (3) to [out=150,in=-80] (2);
\draw [->,blue] (3) -- (0);
\draw [->,blue] (7) -- (3);

\draw [->,ultra thick, dashed, red] (0) -- (1);
\draw [->,ultra thick, dashed, red] (1) -- (2);
\draw [->,green!50!black, dotted] (2) to [out=-135,in=135] (1);
\draw [->,green!50!black, dotted] (1) to [out=-135,in=135] (0);

\draw [->,ultra thick, dashed, red] (5) -- (6);
\draw [->,ultra thick, dashed, red] (6) -- (7);
\draw [->,green!50!black, dotted] (7) to [out=-45,in=45] (6);
\draw [->,green!50!black, dotted] (6) to [out=-45,in=45] (5);

\draw [->,blue] (5) -- (0);
\draw [->,blue] (6) -- (4);
\draw [->,ultra thick, dashed, red] (7) -- (2);

        \end{tikzpicture}
    }
\end{center}
\caption{An example of a DFA
represented as an edge-labeled directed
multigraph.
Loops are not drawn, and arcs with the same label receive the same colour.
The labels corresponding to the colours are shown.
}
\end{figure}

DFAs give a framework for representing a given machine, along with possible transitions between the machine's states. 
When using a machine with multiple states, it is often convenient to have a command that resets the machine from any state to some initial state. For instance, a machine's current state may become unknown
to a user, say after the machine receives accidental or random input. In these cases, a reset command that updates the machine to a prescribed state allows a user to regain control over the machine state and execute meaningful commands. 
Given a DFA $M = (Q,\Sigma,\delta)$ and a state
$z \in Q$, if there exists a word $\tau$ such that $q(\tau) = z$ for all $q \in Q$, then we say that $\tau$ is a \emph{$z$-synchronizing word}. If $M$ has a $z$-synchronizing word, then we say $M$ is \textit{$z$-synchronizable}. If the choice of $z$ is clear from context or is not relevant to a specific application, then  we say that $\tau$ is a \textit{synchronizing word} in $M$ and that $M$ is \textit{synchronizable}.
In addition to having direct machine applications, synchronizing words also have applications in coding theory; for more details,
see the thesis of Ryzhikov \cite{Ryzhikov}.

The formulation of synchronizing words 
stated above
was first considered in the early 1960s, most notably by \v{C}ern\'y in \cite{vcerny1964poznamka}, with the notion appearing first slightly earlier in \cite{laemmel1956general,laemmel1963study,liu1962some}. In \cite{vcerny1964poznamka} \v{C}ern\'y constructed, for each $n$, a DFA of order $n$ whose shortest synchronizing word has length exactly $(n-1)^2$. This construction naturally leads to the following conjecture, often called \v{C}ern\'y's conjecture, seemingly first proposed by Starke \cite{starke1966bemerkung}.
For more history on this conjecture, see \cite{volkov2022synchronization}.

\begin{conjecture}[\v{C}ern\'y's conjecture]
    If $M$ is a synchronizable DFA of order $n$, then $G$ has a synchronizing word of length at most $ 
    (n-1)^2$.
\end{conjecture}

Since being proposed, \v{C}ern\'y's conjecture has become a central open question in automata theory. Until recently, the most notable progress on \v{C}ern\'y's conjecture came from work of Pin \cite{pin1983two} and Frankl \cite{frankl1982extremal}, who showed that each synchronizable DFA of order $n$ has
a synchronizing word of length at most $(n^3-n)/6$. The current best upper bound on the length of a shortest synchronizing word in a synchronizable DFA of order $n$ is $(\frac{7}{48} + \frac{15625}{798768})n^3+o(n^3)$,
as shown by Shitov \cite{shitov2019improvement}. Shitov's improvement builds off two earlier papers by Szyku\l{}a \cite{szykula2018improving} and Trahtman \cite{trahtman2011modifying}. 
Thus, for a general synchronizable DFA, the best upper bound remains far from \v{C}ern\'y's conjecture, despite significant effort.

Notably, \v{C}ern\'y's conjecture has also been proven for some specific classes of DFAs. These include Eulerian DFAs \cite{kari2003synchronizing}, DFAs whose digraphs have
a Hamiltonian cycle in which every edge has the same label \cite{dubuc1998automates}, and DFAs which admit particular
orderings on states that are preserved under the action of any word \cite{ananichev2004synchronizing,ananichev2005synchronizing,eppstein1990reset}.
Another class of DFAs which follow \v{C}ern\'y's conjecture is the class of aperiodic DFAs.
A monoid $S$ is \emph{aperiodic} if for each $s \in S$, there exists an integer $m\geq 1$ such that $s^{m+1} = s^m$.
A DFA $M = (Q,\Sigma, \delta)$ is \emph{aperiodic}
if its transition monoid is aperiodic---that is, if
there exists an integer $m \geq 1 $ such that
for all $q \in Q$ and $\tau\in \Sigma^*$, $q(\tau^m) = q(\tau^{m+1})$.
Trahtman \cite{trahtman2007v}
showed that an $n$-state aperiodic DFA $M$ is synchronizable if and only if it has a universally reachable state,
and he showed that $M$ has a synchronizing word of length at most $\frac 12 n(n-1)$.
For more on
\v{C}ern\'y's conjecture and synchronizing DFAs more broadly, we recommend the 2022 survey by Volkov \cite{volkov2022synchronization}.

In this paper, we define an intuitive strategy for generating a synchronizing word in a synchronizable DFA.
We call our strategy the \emph{cornering strategy},
as it takes advantage of a state with similar properties to those of a vertex in a polytope.
We show that we can construct a synchronizing word $\tau$ for every synchronizable DFA $M$ using the cornering strategy and that $\tau$ is short when $M$ is well structured.
We show that the cornering strategy is particularly well suited for the class of \emph{difference DFAs},
which serve as a natural model for movement in real space.
In addition to applying the cornering strategy to difference DFAs, we also investigate the relationship between difference DFAs and other DFA classes.
In particular, we show that 
while difference DFAs are not aperiodic in general,
they share properties with aperiodic DFAs
that imply that they are synchronizable whenever they have a universally reachable state.
Using similar ideas to those of the cornering strategy,
we develop the \emph{$f$-ordering} strategy, which produces a synchronizing word $\tau$ for every synchronizable DFA $M$ using a certain partial order on the states of $M$. Furthermore, we show that $\tau$ is short whenever the partial order is well structured.

The paper is structured as follows.
In Section~\ref{sec:prelim}, we establish the preliminaries necessary for the rest of the paper, including definitions and examples.
In Section \ref{sec:diff},
we introduce the class of difference DFAs,
and 
we show that difference DFAs share a close connection with aperiodic DFAs.
 Next, in Section~\ref{sec:corner}, we use the cornering strategy to prove that 
 well-structured difference DFAs have short synchronizing words, and we give conditions for these DFAs to satisfy \v{C}ern\'y's conjecture. 
In Section~\ref{sec:f-ordered}, we develop the \emph{$f$-ordering strategy},
and we apply this tool to several existing DFA classes.
 Finally,
Section~\ref{sec:product} focuses on generating a word $\tau$ which is $z_1$-synchronizing in a DFA $M_1 = (Q_1, \Sigma,\delta_1)$ and $z_2$-synchronizing in a DFA $M_2 = (Q_2,\Sigma,\delta_2)$, where the states $z_1$ and $z_2$ are given beforehand. This problem is natural given situations where two different models of machines 
receive
radio commands over the same channel.

\section{Preliminaries}
\label{sec:prelim}
This section establishes definitions and notation which are vital to the remainder of the paper.
Following Eppstein
\cite{eppstein1990reset},
we typically use the symbol $\tau$ to denote words,
and we also use the symbols $\pi$ and $\alpha$ when necessary.
Similar to Trahtman \cite{trahtman2007v},
we write $q(\tau)$ for
the image of a state $q$ under the action of a word $\tau$.
We typically use the symbol $z$ to denote a state to which the action of a synchronizing word maps all states in DFA.
We often use bold symbols to represent vectors in $\mathbb R^d$.

Given a DFA $M = (Q, \Sigma, \delta)$,
we say that a state $p \in Q$ is \emph{reachable} from a state $q \in Q$ if there exists a word $\tau \in \Sigma^*$ for which $q(\tau) = p$.
We say that $p$ is \emph{universally reachable}
if $p$ is reachable from $q$ for every $q \in Q$.
Trahtman \cite{trahtman2007v} uses the term \emph{sink} for a universally reachable state, but we avoid this term due to its alternative meaning in graph theory.
We say that
$M$ is \textit{strongly connected} if
every state of $M$ is universally reachable.
We observe that $M$ is strongly connected if
and only if its underlying digraph is strongly connected. 
 For an example of a DFA  that is not strongly connected, see Figure~\ref{Fig:DFA example}.
 For an example of a strongly connected DFA,
 see Figure~\ref{Fig:Diff DFA}. 
It is sometimes useful to have an even stronger notion of connectivity in a DFA. To this end, we say a DFA $M$
with a digraph representation
$G = (V,E,\psi)$ is \textit{bidirectional} if for every edge $(p,q) \in E$, there exists an edge $(q,p)\in E$. 
Then, we say $M$ is \emph{bidirectional connected} if $M$ is bidirectional and strongly connected.

Given two DFAs $M_1 = (Q_1, \Sigma, \delta_1,s_1,A_1)$ and $M_2 = (Q_2, \Sigma, \delta_2,s_2,A_2)$  on the same alphabet, the \emph{product DFA} of $M_1$ and $M_2$, written $M_1 \times M_2$, 
is the DFA $(Q_1 \times Q_2, \Sigma, \delta, (s_1,s_2), A_1 \times A_2)$, where for each $(q_1, q_2) \in Q_1 \times Q_2$ and $c \in \Sigma$, $\delta((q_1,q_2), c) = (\delta(q_1,c), \delta (q_2,c))$.
We note that an active state of $M_1 \times M_2$ represents an active state from each of $M_1$ and $M_2$. In particular, for all $p \in Q_1$, $q \in Q_2$, and $\tau\in \Sigma^*$, $(p,q)(\tau) = (p(\tau),q(\tau))$.

The notion of an isomorphism between DFAs is defined as follows.

\begin{definition}
    Let $M_1 = (Q_1, \Sigma_1, \delta_1, s_1, A_1)$ and $M_2 = (Q_2, \Sigma_2, \delta_2, s_2, A_2)$ be two DFAs. We say that $M_1$ and $M_2$ are \emph{isomorphic} if there exists
    bijections $\phi: Q_1 \rightarrow Q_2$ and $\phi':\Sigma_1 \rightarrow \Sigma_2$ such that the following hold:
    \begin{enumerate}
        \item For each $q \in Q_1$ and $c \in \Sigma_1$, $\phi(\delta_1(q,c)) = \delta_2 (\phi(q), \phi'(c))$;
        \item $\phi(s_1) = s_2$;
        \item $\phi(A_1) = A_2$.
    \end{enumerate}
    In the case that the starting state and accepting states of $M_1$ and $M_2$ are not specified, then we only require condition (1) to hold.
\end{definition}

The following well-known theorem appears as Theorem 5.5 of \cite{Rich}.
\begin{theorem}
\label{thm:unique-min}
    If $L$ is a regular language, then there exists a unique minimal DFA $D$ up to isomorphism for which $L(D) = L$.
\end{theorem}

\begin{figure}\label{Fig:Diff DFA}
\begin{center}
    \scalebox{0.85}{
        \begin{tikzpicture}[node distance={15mm}, thick, main/.style = {draw, circle}]

\node[main][fill= white,draw=white] (z) at (10,0.7) {$(0,1)$};
\node[main][fill= white,draw=white] (z) at (12,0.7) {$(0,-1)$};
\node[main][fill= white,draw=white] (z) at (15,2.75) {$(1,0)$};
\node[main][fill= white,draw=white] (z) at (15,0.75) {$(-1,0)$};

\node[main][fill= white,draw=white] (z2) at (10,3) {};
\node[main][fill= white,draw=white] (z1) at (10,1) {};
\node[main][fill= white,draw=white] (y1) at (12,3) {};
\node[main][fill= white,draw=white] (y2) at (12,1) {};
\node[main][fill= white,draw=white] (x1) at (14,2.25) {};
\node[main][fill= white,draw=white] (x2) at (16,2.25) {};
\node[main][fill= white,draw=white] (w1) at (14,1.25) {};
\node[main][fill= white,draw=white] (w2) at (16,1.25) {};

\draw [->,blue] (z1) -- (z2);
\draw [->,ultra thick, dashed, red] (y1) -- (y2);
\draw [->,green!50!black, dotted] (x1) -- (x2);
\draw [->,cyan, ultra thick, densely dotted] (w2) -- (w1);

\node[main][fill= black] (0) at (0,0) {}; 
\node[main][fill= black] (1) at (0,2) {}; 
\node[main][fill= black] (2) at (0,4) {}; 

\node[main][fill= black] (3) at (2,0) {}; 
\node[main][fill= black] (4) at (2,2) {}; 
\node[main][fill= black] (5) at (2,4) {}; 

\node[main][fill= black] (6) at (4,0) {}; 
\node[main][fill= black] (7) at (4,2) {}; 
\node[main][fill= black] (8) at (4,4) {};

\node[main][fill= black] (9) at (6,0) {}; 
\node[main][fill= black] (10) at (6,2) {}; 
\node[main][fill= black] (11) at (6,4) {}; 

\node[main][fill= black] (12) at (8,0) {}; 
\node[main][fill= black] (13) at (8,2) {}; 
\node[main][fill= black] (14) at (8,4) {}; 

\draw [->,blue] (0) to [out=135,in=-135] (1);
\draw [->,ultra thick, dashed, red] (1) -- (0);
\draw [->,blue] (1) to [out=135,in=-135] (2);
\draw [->,ultra thick, dashed, red] (2) -- (1);

\draw [->,green!50!black, dotted] (0) -- (3);
\draw [->,cyan, ultra thick, densely dotted] (3) to [out=-135,in=-45](0);

\draw [->,green!50!black, dotted] (1) -- (4);

\draw [->,green!50!black, dotted] (2) -- (5);
\draw [->,cyan, ultra thick, densely dotted] (5) to [out=135,in=45](2);

\draw [->,blue] (3) to [out=135,in=-135] (4);
\draw [->,ultra thick, dashed, red] (4) -- (3);
\draw [->,blue] (4) to [out=135,in=-135] (5);
\draw [->,ultra thick, dashed, red] (5) -- (4);

\draw [->,green!50!black, dotted] (3) -- (6);
\draw [->,cyan, ultra thick, densely dotted] (6) to [out=-135,in=-45](3);

\draw [->,green!50!black, dotted] (4) -- (7);

\draw [->,green!50!black, dotted] (5) -- (8);
\draw [->,cyan, ultra thick, densely dotted] (8) to [out=135,in=45](5);

\draw [->,blue] (6) to [out=135,in=-135] (7);
\draw [->,ultra thick, dashed, red] (7) -- (6);
\draw [->,blue] (7) to [out=135,in=-135] (8);
\draw [->,ultra thick, dashed, red] (8) -- (7);

\draw [->,green!50!black, dotted] (6) -- (9);
\draw [->,cyan, ultra thick, densely dotted] (9) to [out=-135,in=-45](6);

\draw [->,cyan, ultra thick, densely dotted] (10) -- (7);

\draw [->,green!50!black, dotted] (8) -- (11);

\draw [->,blue] (9) to [out=135,in=-135] (10);
\draw [->,ultra thick, dashed, red] (10) -- (9);
\draw [->,blue] (10) to [out=135,in=-135] (11);
\draw [->,ultra thick, dashed, red] (11) -- (10);

\draw [->,cyan, ultra thick, densely dotted] (12) -- (9);

\draw [->,cyan, ultra thick, densely dotted] (13) -- (10);

\draw [->,green!50!black, dotted] (11) -- (14);
\draw [->,cyan, ultra thick, densely dotted] (14) to [out=135,in=45](11);

\draw [->,blue] (12) to [out=135,in=-135] (13);
\draw [->,ultra thick, dashed, red] (13) -- (12);
\draw [->,blue] (13) to [out=135,in=-135] (14);
\draw [->,ultra thick, dashed, red] (14) -- (13);

        \end{tikzpicture}
    }
\end{center}
\caption{A digraph representation $G$ of a difference DFA with a vertex set $V(G) \subseteq \mathbb{R}^2$, with no loops drawn and with edges with the same label receiving the same colour. In this example, the vertex set $V$ is equal to the set of integer points in the rectangle $[0,4] \times [0,2]$.}
\end{figure}

Given a DFA $M = (Q,\Sigma,\delta)$
with a digraph
$G$ and a  pair of states $p,q \in Q$, we write $\dist(p,q)$ for the number of edges in a
shortest directed path in $G$
from $p$ to $q$. From a graph theoretic perspective,
$\dist(p,q)$ is the ordinary directed graph distance.
If there is no walk in $G$ from $p$ to $q$, then we write $\dist(p,q) = \infty$.
Notice that if $M$ is a bidirectional connected DFA, then for all $p,q \in Q$, $\dist(p,q)=\dist(q,p)$.

\section{Difference DFAs and Geodesically Aperiodic DFAs}
\label{sec:diff}


One problem that is easily modeled by DFAs is movement in real life and simulated spaces. DFAs with vertices in $\mathbb{R}^2$ and $\mathbb{R}^3$
are particularly suited 
for modeling the movement of robots and simple machines in a grid or some other discrete space.
In these applications, a synchronizing word is often useful.
For example, given a delivery robot in a warehouse whose movement is modeled by a DFA with vertices in $\mathbb R^2$,
a synchronizing word corresponds to 
a universal command that returns the robot from any location to 
a fixed point, such as a loading dock or charging port.
Such a universal command is useful for simple machines that are unable to relay data to the controller about their current location or state. 
Furthermore,
the inclusion of synchronizing words throughout a sequence of commands prevents a single incorrect input from spoiling 
all subsequent commands.

Another related application comes from taking our DFA to be a subgraph of a grid.
In such instances we can model how a robot moves through a maze. 
In particular, if such a maze DFA has a synchronizing word, then this 
allows
the robot to deterministically solve the maze without any knowledge of its current (or initial) position.

Not every maze or warehouse can be modeled using a rectangular grid. 
For instance, in a two-dimensional maze on a triangular grid or a rectangular grid in three dimensions, up to six paths can meet at a common point.
No such maze can be modeled using a square grid, where intersections are vertices and edges are paths between intersections, as square grids have maximum degree $4$.
Hence, it is worth considering the same problems for DFAs that are not subgraphs of grids.

We formalize this notion
as follows.
A \emph{difference DFA}
$M = (Q,\Sigma,\delta)$ 
in $\mathbb R^d$
is defined by letting $Q \subseteq \mathbb R^d$,
$\Sigma \subseteq \{\mb q - \mb p : \mb p, \mb q \in Q\}$,
and $\delta(\mb q, \mb c) \in \{\mb q, \mb q + \mb c\} \cap Q$ for each $\mb q \in Q$ and $\mb c \in \Sigma$.
We require each vector in $\Sigma$ to be nonzero.
One may imagine that each symbol $\mb c \in \Sigma$ commands a robot
to move along a translation vector $\mb c$ and that the robot obeys the command if and only if the translation vector $\mb c$ can be traversed without encountering an obstacle.
In this way, the digraph representation $G = (V,E,\psi)$ of $M$ has a vertex set $V(G) = Q$ that is a subset of $\mathbb R^d$.
Furthermore, for each ordered pair $\mb p, \mb q \in V(G)$,
if $E(G)$ contains an arc $e = (\mb p, \mb q)$, then
$\mb q - \mb p \in \Sigma$ and 
$\psi(e) = \mb q - \mb p$.
Otherwise, either $\mb q - \mb p \not \in \Sigma$,
or 
$E$ contains a loop $e' = (\mb p, \mb p)$ satisfying $\psi(e') = \mb q - \mb p$.
 Notice that $M$ 
 need not be bidirectional or strongly connected. For an example of a difference DFA, see Figure~\ref{Fig:Diff DFA}.

In this section, we explore the relationship between difference DFAs and 
other existing classes of DFAs.
In particular, we will see that difference DFAs have a close relationship with \emph{aperiodic DFAs}. While difference DFAs are not aperiodic in general, they share important properties with aperiodic DFAs
that can help us analyze their synchronizing words.

Recall that a DFA $M = (Q,\Sigma, \delta)$ is \emph{aperiodic}
if there exists an integer $m \geq 1 $ such that
for all $q \in Q$ and $\tau\in \Sigma^*$, $q(\tau^m) = q(\tau^{m+1})$.
It is easy to construct examples of difference DFAs that are not aperiodic.
In
Figure~\ref{Fig: Diff not aperiodic} we provide an example of a difference DFA in $\mathbb{R}^2$ that is not  aperiodic. 
Interestingly, the example is bidirectional connected and is a subgraph of the $3\times 3$ square grid. We highlight the case of subgraphs of square grids as it is of natural interest.

\begin{figure}[h]\label{Fig: Diff not aperiodic}
\begin{center}
    \scalebox{0.85}{
        \begin{tikzpicture}[node distance={15mm}, thick, main/.style = {draw, circle}]

\pgfmathsetmacro{\a}{3}

\node[main][fill= white,draw=white] (z) at (10-\a,0.7) {$(0,1) = \uparrow$};
\node[main][fill= white,draw=white] (z) at (12-\a,0.7) {$(0,-1) = \downarrow$};
\node[main][fill= white,draw=white] (z) at (15-\a,2.75) {$(1,0) = \rightarrow$};
\node[main][fill= white,draw=white] (z) at (15-\a,0.75) {$(-1,0) = \leftarrow$};

\node[main][fill= none,draw=none] (z2) at (10-\a,3) {};
\node[main][fill= none,draw=none] (z1) at (10-\a,1) {};
\node[main][fill= none,draw=none] (y1) at (12-\a,3) {};
\node[main][fill= none,draw=none] (y2) at (12-\a,1) {};
\node[main][fill= none,draw=none] (x1) at (14-\a,2.25) {};
\node[main][fill= none,draw=none] (x2) at (16-\a,2.25) {};
\node[main][fill= none,draw=none] (w1) at (14-\a,1.25) {};
\node[main][fill= none,draw=none] (w2) at (16-\a,1.25) {};

\draw [->,blue] (z1) -- (z2);
\draw [->,ultra thick, dashed, red] (y1) -- (y2);
\draw [->,green!50!black, dotted] (x1) -- (x2);
\draw [->,cyan, ultra thick, densely dotted] (w2) -- (w1);

\node[fill=none] at (2,2.5,0) (nodes) {$q$};
\node[fill=none] at (2,4.5) (nodes) {$p$};
\node[main][fill= black] (0) at (0,0) {}; 
\node[main][fill= black] (1) at (2,0) {}; 
\node[main][fill= black] (2) at (4,0) {}; 
\node[main][fill= black] (3) at (0,2) {}; 
\node[main][fill= black] (4) at (2,2) {}; 
\node[main][fill= black] (5) at (4,2) {}; 
\node[main][fill= black] (6) at (0,4) {}; 
\node[main][fill= black] (7) at (2,4) {};
\node[main][fill= black] (8) at (4,4) {};

\draw [->,green!50!black, dotted] (0) -- (1);
\draw [->,green!50!black, dotted] (1) -- (2);
\draw [->,cyan, ultra thick, densely dotted] (1) to [out=135,in=45](0);
\draw [->,cyan, ultra thick, densely dotted] (2) to [out=135,in=45](1);
\draw [->,green!50!black, dotted] (6) -- (7);
\draw [->,green!50!black, dotted] (7) -- (8);
\draw [->,cyan, ultra thick, densely dotted] (7) to [out=135,in=45](6);
\draw [->,cyan, ultra thick, densely dotted] (8) to [out=135,in=45](7);

\draw [->,blue] (0) to [out=135,in=-135] (3);
\draw [->,blue] (3) to [out=135,in=-135] (6);
\draw [->,blue] (2) to [out=135,in=-135] (5);
\draw [->,blue] (5) to [out=135,in=-135] (8);
\draw [->,blue] (1) to [out=135,in=-135] (4);
\draw [->,ultra thick, dashed, red] (3) -- (0);
\draw [->,ultra thick, dashed, red] (6) -- (3);
\draw [->,ultra thick, dashed, red] (5) -- (2);
\draw [->,ultra thick, dashed, red] (8) -- (5);
\draw [->,ultra thick, dashed, red] (4) -- (1);

        \end{tikzpicture}
    }
\end{center}
\caption{An example of a bidirectionally connected difference DFA in $\mathbb{R}^2$ whose digraph is a subgraph of a the square lattice, and that is not an aperiodic DFA. To see that this DFA is not an aperiodic DFA,
observe that when $\tau = (\leftarrow,\downarrow, \rightarrow, \downarrow, \rightarrow, \uparrow , \uparrow, \leftarrow)$, $p(\tau) = q$ and $q(\tau) = p$. As a result, the action of $\tau$ cyclically permutes $p$ and $q$,
implying this DFA is not aperiodic.}
\end{figure}

In addition to showing that certain difference DFAs are not aperiodic, we can in fact show that 
every DFA in the class of \emph{unitary DFAs} is isomorphic to a difference DFA.
Since many unitary DFAs are not aperiodic, this 
fact gives us many examples of difference DFAs that are not aperiodic.
Unitary DFAs were introduced by  Brzozowski and Szyku\l{}a \cite{Brz} and defined as follows:
A DFA $M = (Q,\Sigma, \delta)$ is \emph{unitary} if for each $c \in \Sigma$,
$\delta(q,c) = q$ for at least $|Q| - 1$ elements $q \in Q$.
In other words, if $\delta$ gives $c \in \Sigma$ a non-identity action on $Q$, then there is a unique element $q \in Q$ for which $\delta(q,c) \neq q$.
Writing $\delta(q,c ) = p$,
we write the function $\delta(\cdot, c)$ as $q \rightarrow p$.
For an example of a unitary DFA that is not aperiodic, consider the DFA $M = (Q,\Sigma,\delta)$, where $Q = \{p,q,r\}$ and $\{\delta(\cdot, c): c \in \Sigma\}$ consists of the functions $p \rightarrow q$, $q \rightarrow r$, and $r \rightarrow p$.
Then, the word $\tau \in \Sigma^*$ corresponding to the composition of $q \rightarrow r, p \rightarrow q, r \rightarrow p$ (applied from left to right) cyclically permutes the states $p$ and $q$, and hence $p(\tau^m) \neq p(\tau^{m+1})$ for all $m \geq 1$.

\begin{proposition}
    Every unitary DFA is isomorphic to a difference DFA. 
\end{proposition}
\begin{proof}
    Let $M = (Q, \Sigma, \delta)$ be a unitary DFA.
    We construct an isomorphic
    difference DFA in the form of an
    edge-labeled 
    digraph $G = (V,E,\psi)$ 
    as follows.
    
    Let $V \subseteq \mathbb N$ be a Sidon set of size $|Q|$---that is, a set with the property that if $p_1, p_2, q_1, q_2 \in Q$ satisfy $q_2 - q_1 = p_2 - p_1$, then $(p_1, p_2) = (q_1,q_2)$. 
    One possible choice for $V$ is $\{2^k : 0 \leq k \leq |Q| - 1\}$. 
    Fix some bijection $\phi: Q\rightarrow V$.
    We construct an alphabet $\Sigma'$
    for $G$ and a bijection $\phi':\Sigma \rightarrow \Sigma'$ as follows.
    For each $c \in \Sigma$
    for which $\delta( \cdot, c) = p \rightarrow q$,
    let $G$ have an arc $e = (\phi(p),\phi(q))$ with label $\psi(e) = \phi(q)-\phi(p)$.
    Add $\psi(e)$ to $\Sigma'$, and define $\phi'(c) = \psi(e)$.
    For each remaining $c \in  \Sigma$ for which $\delta(\cdot, c)$ is the identity transformation, let $x$ be some
    real number that is not yet in $\Sigma'$ 
    and that cannot be obtained as a difference of two values in $V$.
    In particular, 
    every non-integer $x \not \in \Sigma'$ is a valid choice.
    Add $x$ to $\Sigma$, and define $\phi'(c) = x$.
    Since $V$ is a Sidon set, each element $\psi(e)$ added to $\Sigma'$ is distinct, and hence $\phi'$ is a bijection.
    It is easy to see that the bijections $\phi$ and $\phi'$ give an isomorphism between $M$ and the DFA defined by $G$.
\end{proof}

Next, we show that the language of every DFA whose transition monoid is commutative and aperiodic can be represented by a difference DFA.

\begin{proposition}
\label{prop:comm-ap}
    If $M = (Q,\Sigma,\delta,s,A)$ is an
    aperiodic DFA whose transition monoid is commutative, then
    there exists a difference DFA $D$ for which $L(M) = L(D)$ after relabeling the elements of $\Sigma$.
\end{proposition}
\begin{proof}
    Suppose that $M = (Q,\Sigma, \delta,s,A)$ is a DFA whose transition monoid is commutative and aperiodic.
    Let $m \geq 1$ be an integer such that $q(\tau^{m+1}) = q(\tau^m)$ for each $q \in Q$ and $\tau \in \Sigma^*$.
    Given a word $\tau \in \Sigma^*$ and a symbol $c \in \Sigma$, 
    let $t_c$ be the number of times $c$ appears in $\tau$.
    We define the \emph{type} of $\tau$
    as a vector $(k_c)_{c \in \Sigma}$, where $k_c = t_c$ if $t_c \leq m$ and 
    $k_c = m$ if $t > m_c$.
    As $M$ has a commutative transition monoid, the question of whether $M$ accepts or rejects $\tau$ is uniquely determined by the type of $\tau$. 
    
    Now, we construct a difference DFA $D = (Q_{D}, \Sigma', \delta',\mb s_{D},A')$ 
    as follows.
    We let $\mathbb R^{\Sigma}$ be a vector space whose standard basis vectors $\mb e_c$ are indexed by the elements $c \in \Sigma$.
    We let $Q_{D}$ be the set of integer points in $[0,m]^{|\Sigma|}$, and we let $\mb s_{D} = \mb 0$.
    We let $\Sigma' = \{\mb e_c: c \in \Sigma\}$.
    For each vector $\mb q \in Q'$ and $\mb e_c \in \Sigma'$,
    we let $\delta'(\mb q,\mb e_c) = \mb q + \mb e_c$ if $\mb q + \mb e_c \in Q'$; otherwise, we let $\delta'(\mb q, \mb e_c) = \mb q$.
    Finally, for each type $(k_c)_{c \in \Sigma}$ whose words $M$ accepts, we let $\sum_{c \in \Sigma} k_c \mb e_c \in A'$.
    This completes the construction of $D$.

    We claim that $M$ accepts $\tau = c_1\cdots c_r$ if and only if $D$ accepts $\tau' = \mb e_{c_1} \cdots \mb e_{c_r}$.
    Indeed, let $\tau \in \Sigma^*$.
    If $\tau \in L(M)$,
    then
    $M$ accepts words of the same type as $\tau$.
    We write $(k_c)_{c \in \Sigma}$ for the type of $\tau$.
    By construction,  $\sum_{c \in \Sigma} k_c \mb e_c \in A'$,
    so $D$ accepts $\tau' = \mb e_{c_1}, \dots, \mb e_{c_r}$.
    A similar argument shows that if $\tau \not \in L(M)$, then $D$ rejects $\tau'$. Thus, the claim holds.
    Therefore, when the elements of $\Sigma$ are relabeled using the bijection $\phi:\Sigma \rightarrow \Sigma'$ that maps $c \mapsto \mb e_c$,
    we have $L(M) = L(D)$.
\end{proof}

Next, we show that the language of a strongly connected difference DFA can be represented using a minimal difference DFA.
Using this fact, we establish a condition that limits the DFAs whose languages can be represented by a difference DFA.

\begin{proposition}
\label{prop:diff-min-SC} 
If $D = (Q,\Sigma,\delta,\mb s,A)$ is a strongly connected
difference DFA, then $L(D) = L(D')$ for some minimal
difference DFA $D'$.
\end{proposition}
\begin{proof}
    Suppose that $D = (Q,\Sigma,\delta,\mb s,A)$ is a strongly connected difference DFA.
    If every pair of states $\mb p, \mb q \in Q$ is distinguishable (so that exactly one of $\mb p(\tau), \mb q(\tau)$ belongs to $A$ for some $\tau \in \Sigma^*$), then $D$ is minimal. Therefore, we assume that $Q$ contains at least two distinct states that are indistinguishable. We define an equivalence relation $\sim$ so that for each $\mb q_1, \mb q_2 \in Q$, $\mb q_1 \sim \mb q_2$ if $\mb q_1$ and $\mb q_2$ are indistinguishable. 
    We partition $Q$ into equivalence classes using the relation $\sim$.
    By considering $\delta(\cdot, \epsilon)$, we observe that each equivalence class $K$ intersects $A$ either in $K$ or $\emptyset$.

    \begin{claim}
    \label{claim:shift}
        Suppose that $K_1$ and $K_2$ are two equivalence classes in $Q$ and that $\mb q' \in K_1$ and $\mb p' \in K_2$ satisfy $\delta(\mb q', \mb c) = \mb p'$ for some $\mb c \in \Sigma$. Then,
        $\delta(\mb q, \mb c) = \mb q + \mb c \in K_2$ for all $\mb q \in K_1$. Furthermore,
        $K_2 = K_1 + \mb c$.
    \end{claim}
    \begin{proof}
        First, we observe that for each $\mb q \in K_1$, $\delta(\mb q, \mb c) \in K_2$. 
        Indeed, if $\delta(\mb q, \mb c) \not \in K_2$,
        then $\delta(\mb q, \mb c)$ and $\delta(\mb q', \mb c)$ belong to distinct equivalence classes. This implies that for some $\tau \in \Sigma^*$, exactly one of $\delta(\mb q, \mb c \tau)$ and $\delta(\mb q', \mb c \tau)$ is accepted by $D$, contradicting the indistinguishability of $\mb q$ and $\mb q'$.

        Since $\delta(\mb q,\mb c) \in K_2$ for each $\mb q \in K_1$, it follows that $\delta(\mb q, \mb c) \neq \mb q$ for each $\mb q \in K_1$. Therefore, for each $\mb q \in K_1$, the definition of a difference DFA implies that $\delta(\mb q, \mb c) =\mb  q + \mb c$, implying that $\mb q+\mb c \in K_2$.
        Therefore, $K_2$ contains the shifted set $K_1 + \mb c$ as a subset.

        Now, 
        let $G$ be the edge-labeled digraph that represents $D$. As $D$ is strongly connected, $G$ is strongly connected as well. Next, let $H$ be the quotient graph obtained from $G$ by contracting each equivalence class to a single vertex. 
        It is easy to see that $H$ is also strongly connected. Furthermore, if some vertex of
        $H$ corresponding to the class $K_1$ has an out-neighbour corresponding to a class $K_2 \neq K_1$,
        then 
        the fact that $K_2 \supseteq K_1 + \mb c$
        implies that $|K_2| \geq |K_1|$.
        As $H$ is strongly connected, this implies that all equivalence classes contain the same number of states. Therefore, $K_2 = K_1 + \mb c$.
        In other words, for some set $S \subseteq Q$,
        every equivalence class is a translation of the set $S$.
    \end{proof}

    Now, define $D' = (Q',\Sigma, \delta',\mb s', A')$ as follows. 
    Let
    $Q'$ be obtained from $Q$ by replacing each equivalence class $K$ with the single point at the arithmetic mean of the points in $K$.
    For each $\mb q' \in Q'$ and $\mb c \in \Sigma$, let $\delta'(\mb q', \mb c) = \mb p'$ if and only if $\delta'(\mb q,\mb c) = \mb p$
    for some states $\mb q,\mb p$ in the respective equivalence classes corresponding to $\mb q', \mb p'$. By Claim \ref{claim:shift}, $\delta'$ is well defined and obeys the key property of a difference DFA.
    Let $\mb s'$ be the point obtained from the equivalence class containing $\mb s$, and let $A'$ be the set of points obtained from the equivalence classes which are subsets of $A$.
    We show by induction on $k$ that for each word $\tau = c_1 \cdots c_k$ and state $\mb q \in Q$ corresponding to $\mb q' \in Q'$,
    $\mb q(\tau) \in A$ if and only if $\mb q'(\tau) \in A'$.
    For $k = 0$, we have $\tau = \epsilon$,
    so that $D$ and $D'$ both accept $\tau$ if and only if the equivalence class containing $\mb s$ is a subset of $A$.
    Now, suppose that $k \geq 1$, and write $\tau =\mb c \tau'$ for some $\mb c \in \Sigma$.
    Then, $\mb q(\tau) =\mb  q(\mb c) (\tau')$, and $\mb q'(\tau) = \mb q'(\mb c) (\tau')$.
    If $\delta(\mb q,\mb c) = \mb q$, then Claim \ref{claim:shift} implies that $\delta'(\mb q',\mb c) = \mb q'$, and the conclusion holds by induction. 
    Otherwise, $\delta(\mb q,\mb c) = \mb q+\mb c$, and then Claim \ref{claim:shift} implies that $\delta'(\mb q',\mb c) =\mb q' + \mb c$, and again the conclusion holds by induction. Therefore, $L(D) = L(D')$.
    From the same conclusion, we also see that each distinct pair of states in $D'$ is distinguishable, implying that $D'$ is minimal. This completes the proof.
\end{proof}

\begin{proposition}
    Let $M = (Q,\Sigma,\delta,s,A)$ be a minimal
    strongly connected
    DFA.
    Suppose that $M$ has a symbol $c \in \Sigma$ and a state $q \in Q$ for which the preimage of $q$ under $\delta( \cdot, c)$ contains at least three states.
    Then, for every difference DFA $D$, $L(M) \neq L(D)$.
\end{proposition}
\begin{proof}
    Suppose that the proposition is false.
    Let $M = (Q,\Sigma,\delta,s,A)$ be a minimal strongly connected DFA, and let $D = (Q', \Sigma, \delta', \mb s', A')$ be a difference DFA for which $L(M) = L(D)$.
    Let $M$ and $D$ be chosen so that the number of states of $D$ is minimized. 
    Clearly, every state of $D$ is reachable from $\mb s'$.
    As $L(M) = L(D)$, we assume that 
    $Q' \subseteq \mathbb R^d$ and that 
    the elements of $\Sigma$ are $d$-length vectors
    obtained as the differences of point pairs in $Q'$.
    Let $G$ be an edge-labeled digraph on the vertex-set $Q'$ that represents $D$, and let $Q'' \subseteq Q'$ be a subset for which $G[Q'']$ is strongly connected component that is a sink in the condensation graph of $G$.
    In particular, every directed walk in $G$ beginning in $Q''$ also ends in $Q''$.
    Let $\mb s_0 \in Q''$, and let $\tau_0 \in \Sigma^*$ be a word for which $\mb s'(\tau_0) = \mb s_0$.

    We claim that $Q'' = Q'$. Indeed, let $D'$ be the difference DFA $(Q'', \Sigma, \delta', \mb s_0, A' \cap Q'')$.
    Since no walk beginning in $Q''$ leaves $Q''$, $D'$ is complete. Furthermore, for each $\tau \in \Sigma^*$,
    $\tau_0 \tau \in L(M)$ if and only if $D'$ 
    accepts $\tau$.
    Therefore, letting $M' = (Q,\Sigma,\delta,s(\tau_0),A)$,
    $L(M') = L(D')$. 
    As $M$ is minimal and strongly connected, $M'$ is also minimal and strongly connected.
    As $D$ has a minimum number of states, it follows that $D$ and $D'$ both have the same number of states, implying that $Q'' = Q'$.
    In particular, $D$ is a strongly connected DFA.
    As the number of states of $D$ is minimized, Theorem \ref{thm:unique-min} implies that $D$ is minimal.

    Let $q \in Q$ and $\mb c \in \Sigma$ be such that $\delta(q_1, \mb c) = \delta(q_2, \mb c) = \delta(q_3, \mb c) = q$ for three distinct 
    states $q_1, q_2, q_3, \in Q$.
    As $M$ is strongly connected, for each $i \in \{1,2,3\}$, there exists $\tau_i \in \Sigma^*$ for which $s( \tau_i) = q_i$.
    
    For $i \in \{1,2,3\}$,
    write $\mb x_i = \mb s' (\tau_i)$.
    Consider a distinct pair $i,j \in \{1,2,3\}$.
    By minimality of $M$,
    there exists
    a word $\pi$ for which $M$ accepts $ \tau_i \pi$ and rejects $ \tau_j \pi$.
    Hence, $D'$ accepts $\tau_i \pi$ and rejects $\tau_j \pi$. It follows that $ \delta'(\mb s',\tau _i) \neq \delta'(\mb s', \tau_j) $, or in other words, $\mb x_i \neq \mb x_j$.
    Therefore, $\mb x_1, \mb x_2, \mb x_3$ are all distinct.

    Now, since $D$ is a difference DFA, it follows without loss of generality that $\delta'(\mb x_1, \mb c) \neq \delta'(\mb x_2, \mb c)$.
    Since $D$ is minimal,
    there exists some word $\pi'$ for which 
     $\delta'(\mb x_1, \mb c \pi') \in A'$
     and
     $\delta'(\mb x_2, \mb c \pi') \not \in  A'$.
    Equivalently, $D$ accepts $\tau_1 \mb c \pi'$ and rejects $ \tau_2 \mb c \pi'$.
    This implies that $M$ accepts $ \tau_1 \mb c \pi$ and rejects $ \tau_2 \mb c \pi$.
    However, since $\delta(s, \tau_1 \mb c)= \delta(s, \tau_2 \mb c)
     = q $, this gives us a contradiction.
\end{proof}

Although difference DFAs are not aperiodic in general, they do belong to a larger class of DFAs that share many important properties with aperiodic DFAs, which we define below.
Given a DFA $M= (Q,\Sigma,\delta)$ and a state subset $P \subseteq Q$ with $|P| \geq 2$,
we say that a word $\tau \in \Sigma^*$
\emph{cyclically permutes $P$}
if the action of $\tau$ on $P$
is a
cyclic permutation of the states of $P$---that is,  $\delta( \cdot, \tau)$ maps 
$(p_1, \dots, p_{|P|})$ to $(p_2, \dots, p_{|P|}, p_1)$ when the states of $P$ are appropriately labeled as $p_1, \dots, p_{|P|}$.
Note that $M$ is aperiodic if and only if no word $\tau$ exists for which $\tau$ cyclically permutes some non-singleton set $P \subseteq Q$.

Given a digraph $G$ and two vertices $p,q \in V(G)$,
we say that a walk $W$ in $G$ from $p$ to $q$ is a \emph{geodesic path} if the number of edges traversed by $W$ is equal to $\dist(p,q)$.
We say that $M$ is \emph{geodesically aperiodic}
if the following holds: For every set
$P = \{p_1, \dots, p_{|P|}\}$ of at least two states,
if the action of a word 
$\tau $ of length $\ell$ maps $(p_1, \dots, p_{|P|})$ to $(p_2, \dots, p_{|P|}, p_1)$, then there exists $j \in \{1, \dots, |P|\}$ for which $\dist(p_j,p_{j+1}) < \ell$, with $1$ and $|P| + 1$ identified.
In other words, if the action of $\tau$ maps 
$(p_1, \dots, p_{|P|})$ to $(p_2, \dots, p_{|P|}, p_1)$,
then for some $j \in \{1, \dots, |P|\}$,
the walk in the digraph of $M$ from $p_j$ to $p_{j+1}$ given by $\tau$ is not a geodesic path.

\begin{theorem}
    If $M = (Q,\Sigma,\delta)$ is a difference DFA, then $M$ is geodesically aperiodic.
\end{theorem}
\begin{proof}
    Suppose for the sake of contradiction that
    $P = \{\mb p_1, \dots, \mb p_k\}$ is a subset of $Q$
    for which some word $\tau$
    gives a geodesic path in the digraph of $M$ from each state $\mb p_j$ to $\mb p_{j+1}$, where $k+1 = 1$.
    Then, writing
    $\tau = \mb c_1 \cdots \mb c_m$,
    it holds in particular that 
    $\mb p( \mb c_1 \cdots \mb c_{t-1}) \neq \mb p( \mb c_1 \cdots \mb c_t )$
    for each $\mb p \in P$ and $t \in \{1, \dots, m\}$, where $\mb c_1 \cdots \mb c_0$ is the empty word $\epsilon$.
    Then, since $M$ is a difference DFA, it follows that
    $\mb p (\tau) = \mb p + \mb c_1 + \dots + \mb c_m$ for each $\mb p \in P$.
    Furthermore, as the permutation given by the action of $\tau$ on $P$ is nontrivial, $m \geq 1$.
    Thus, the action of $\delta( \cdot, \tau)$ on $P$ is a nonzero shift in real space, and not a permutation, a contradiction.
\end{proof}

Trahtman \cite{trahtman2007v} showed that every aperiodic DFA with a universally reachable state has a synchronizing word. Using similar ideas, we can prove the same result for geodesically aperiodic DFAs, and hence for difference DFAs in particular.

\begin{theorem}
\label{thm:loopless-sync}
    If $M = (Q,\Sigma,\delta)$ is a geodesically aperiodic DFA,
    then $M$ has a synchronizing word if and only if $M$ has a universally reachable state.
\end{theorem}
\begin{proof}
    If $M$ has a synchronizing word $\tau$, then there exists $z \in Q$ such that $q(\tau) = z$ for every $q \in Q$. Therefore, $z$ is universally reachable.

    On the other hand, suppose that $M$ has a state $z$ that is universally reachable.
    We claim that 
    for each $q \in Q$,
    there exists a word $\tau$ 
    for which
    $q(\tau) = z(\tau)$.
    As $z$ is universally reachable,
    this claim clearly implies that $M$ has a synchronizing word.
    To prove the claim,
    suppose that the claim is false, and let $q \in Q$ be a
    state
    for which the claim does not hold. Let $d$ be the minimum integer for which some word $\tau$ 
    satisfies $\dist(q(\tau), z(\tau)) = d$.
    As no word's action maps $q$ and $z$ to the same state, it follows that $d \geq 1$.

    Now,
    let $\pi = c_1 \cdots c_m$ be a word corresponding to a
    shortest path from $q$ to $z$. 
    As $q(\pi) = z$,
    $z(\pi^k) = q(\pi^{k+1})$ for each $k \geq 0$.
     As $Q$ is finite,
    there exist two distinct integers $k < \ell$ for which $q(\pi^k) = q(\pi^{\ell})$. 
    If $\ell = k+1$, then $q(\pi^k) = q(\pi^{k+1}) = z(\pi^k)$,
    a contradiction. Therefore, $\ell - k \geq 2$.
    Letting $P$ be the ordered vertex set $(q(\pi^k), q(\pi^{k+1}), \dots, q(\pi^{\ell-1}))$, we see that the image
    of $P$ under the action of $\pi$ is 
    $(q(\pi^{k+1}),q( \pi^{k+2} ) , \dots, q(\pi^{\ell-1}) ,q(\pi^k))$, so that the action of $\pi$ gives a cyclic permutation of $P$.
    Therefore, as $M$ is geodesically aperiodic,
    for some $i \in \{k, \dots, \ell-1\}$,
    $\dist(q(\pi^i), z(\pi^i)) = \dist(q(\pi^i), q(\pi^{i+1})) < d$,
    a contradiction.
\end{proof}

While geodesically aperiodic DFAs
with a universally reachable state
have a synchronizing word,  Trahtman's proof  of \v{C}ern\'y's conjecture for aperiodic DFAs unfortunately does not extend to geodesically aperiodic DFAs.
In particular,
Trahtman shows that in an aperiodic DFA $M$, a certain partial order can be defined on the states of $M$ that can be exploited when sychronizing two states.
However, this partial order does not clearly hold for geodesically aperiodic DFAs, leaving the question of whether \v{C}ern\'y's conjecture holds for geodesically aperiodic DFAs open.

By analyzing the proof of Theorem \ref{thm:loopless-sync},
we can obtain an upper bound on the minimum length of a synchronizing word in a geodesically aperiodic DFA $M$ with $n$ states and with a universally reachable state.
For convenience, we perform our analysis in the digraph setting.
We let $G$ be the digraph representing $M$, and we
imagine that we wish to guide $n$ robots occupying the $n$ vertices of $G$ to a common vertex of $G$ using commands in $\Sigma^*$.
Our argument shows that if $R$ and $R_0$ are two 
robots, and if $R_0$ occupies a universally reachable state, then repeated applications of the same word $\pi$ can reduce the distance from $R$ to $R_0$, 
and this process can be repeated until $R$ and $R_0$
occupy the same vertex.
Letting $d$ be the diameter of the digraph of $M$,
we have $|\pi| \leq d$,
and $\pi$ needs to be applied at most $n$ times to reduce 
 the $(R,R_0)$-distance.
Therefore, a word of length at most $d^2 n$ can guide $R$ and $R_0$ to a common universally reachable state. 
By repeating this process for $n-1$ robot pairs, we find a synchronizing word for $M$ of length at most $d^2 n(n-1)$.
While this analysis fails to show that difference DFAs satisfy \v{C}ern\'y's conjecture,
we will see that
the cornering strategy
(Theorem \ref{Thm: Cornering Strategy})
shows
that bidirectional difference DFAs with reasonable diameter constraints do satisfy \v{C}ern\'y's conjecture.

\section{The cornering strategy}
\label{sec:corner}
In this section, we introduce the cornering strategy and then apply it to bidirectional difference DFAs.
First, we establish some definitions related to the cornering strategy.
Given a DFA $M = (Q,\Sigma,\delta)$
and a 
function $f:\Sigma^* \rightarrow \{-1,1\}$,
we say that a word $\tau \in \Sigma^*$ is \textit{$f$-positive} if $f(\tau)=1$. 
Letting  $G = (V,E,\psi)$ be the digraph representation of $M$, we extend the definition of $f$ so that $f(W) := f(\psi(W))$ for each directed walk $W$ in $G$. 
Then,
we say that a walk $W$ in $G$ is \textit{$f$-positive} if $f(W)=1$.
When the choice of $f$ is clear from context,
we simply say that
$W$ is positive. 
In this paper, we always define $f$ so that empty word is $f$-positive, i.e.~$f(\epsilon) = 1$.

We write $\dist^+_f(p,q)$ for the length of a shortest $f$-positive walk from $p$ to $q$ in $G$. 
As $f(\epsilon) = 1$, we observe that $\dist^+_f(q,q) = 0$.
Furthermore, if there is no $f$-positive walk from $p$ to $q$, then we write $\dist^+_f(p,q) = \infty$.

Next, we say that $z \in Q$ is an \textit{$f$-corner} in $M$ (or in $G$) if $\dist^+_f(q,z) < \dist^+_f(z,q)$
for all $q  \in Q \setminus \{z\}$. When $f$ is clear from context,
we say that $z$ is a \emph{corner}. For some examples of corners,
see Figure~\ref{fig:corner-example}. In particular,
notice that $z$ is an $f$-corner,
because for all vertices there is a shortest path to $z$ which is $f$-positive, whereas no path leaving $z$ is
$f$-positive. 
Furthermore, while not every $q \in Q\setminus \{z'\}$ has a shortest path to $z'$ that is $f$-positive,
one can still verify that 
 $\dist^+_g(x,z') < \dist^+_g(z',x)$ for all $q \in Q \setminus \{z'\}$.
 Therefore, $z'$ is a $g$-corner.
Also notice that $z'$ is not an $f$-corner,
and $z$ is not a $g$-corner.

\vspace{0.5cm}
\begin{figure}[!ht]\label{Fig:Non-triv corner}
\begin{center}
    \scalebox{0.85}{
        \begin{tikzpicture}[node distance={15mm}, thick, main/.style = {draw, circle}] 

\node[main][fill= black] (0) at (0,4) {};
    \node[fill=none] at (-1,4) (nodes) {$z'$};
\node[main][fill= black] (1) at (1,0) {}; 
\node[main][fill= black] (2) at (5,0) {}; 
\node[main][fill= black] (3) at (6,4) {}; 
\node[main][fill= black] (4) at (3,6) {}; 
    \node[fill=none] at (3,7) (nodes) {$z$};
 
\node[main][fill= black] (5) at (2,3) {}; 
\node[main][fill= black] (6) at (2.5,1.5) {}; 
\node[main][fill= black] (7) at (3.5,1.5) {};
\node[main][fill= black] (8) at (4,3) {}; 
\node[main][fill= black] (9) at (3,4) {}; 
 
\draw [->,cyan, ultra thick, densely dotted] (4) -- (0);
\draw [->,ultra thick, dashed, red] (0) to [out = 90, in = 180] (4);
\draw [->,green!50!black, dotted] (4) -- (3);
\draw [->,blue] (3) to [out = 90, in = 0] (4);
\draw [->,ultra thick, dashed, red] (9) -- (4);

\draw [->,green!50!black, dotted] (7) -- (5);
\draw [->,blue] (6) -- (9);
\draw [->,ultra thick, dashed, red] (7) -- (9);
\draw [->,blue] (5) -- (0);
\draw [->,blue] (8) -- (3);

\draw [->,ultra thick, dashed, red] (1) -- (0);
\draw [->,green!50!black, dotted] (1) -- (2);
\draw [->,cyan, ultra thick, densely dotted] (1) -- (6);

\draw [->,green!50!black, dotted] (2) to [out = -135, in = -45] (1);
\draw [->,blue] (2) -- (3);
\draw [->,cyan, ultra thick, densely dotted] (2) -- (7);

\draw [->,ultra thick, dashed, red] (5) -- (8);
\draw [->,cyan, ultra thick, densely dotted] (8) -- (6);

\draw [->,cyan, ultra thick, densely dotted] (3) to [out = -45, in = 45] (2);
\draw [->,cyan, ultra thick, densely dotted] (9) to [out = 45, in = -45] (4);

\pgfmathsetmacro{\a}{2}
\node[main][fill= white,draw=white] (z1) at (6+\a,3) {};
\node[main][fill= white,draw=white] (z2) at (6+\a,1) {};
\node[main][fill= white,draw=white] (y1) at (8+\a,3) {};
\node[main][fill= white,draw=white] (y2) at (8+\a,1) {};
\node[main][fill= white,draw=white] (x1) at (10+\a,3) {};
\node[main][fill= white,draw=white] (x2) at (10+\a,1) {};

\node[main][fill= white,draw=white] (z) at (6.5+\a,2) {\texttt{a}};
\node[main][fill= white,draw=white] (z) at (8.5+\a,2) {\texttt{b}};
\node[main][fill= white,draw=white] (z) at (10.5+\a,2) {\texttt{c}};

\draw [->,blue] (z1) -- (z2);
\draw [->,ultra thick, dashed, red] (y1) -- (y2);
\draw [->,green!50!black, dotted] (x1) -- (x2);

        \end{tikzpicture}
    }
\end{center}
\caption{A strongly connected DFA with an $f$-corner $z$ and a $g$-corner $z'$. Here, 
$f:\Sigma^* \rightarrow \{-1,1\}$
is a function in which a nonempty word is positive if and only if it begins with \texttt{a} or \texttt{b}, and
$g:\Sigma^* \rightarrow \{-1,1\}$ is a 
function in which a nonempty word is positive if and only if it begins with \texttt{a} or \texttt{c}.}
\label{fig:corner-example}
\end{figure}

\subsection{The general strategy}

This subsection describes
the cornering strategy, 
which gives an explicit method for constructing a synchronizing word in a DFA with an $f$-corner.
The cornering strategy  is
given in the proof of Theorem~\ref{Thm: Cornering Strategy}. Thus,
Theorem~\ref{Thm: Cornering Strategy} can be viewed as a summary of the necessary conditions for applying
the cornering strategy, as well as an upper bound on the length of the resulting synchronizing word. In Theorem~\ref{Thm: Cornering is necessary}, we prove that a DFA is $z$-synchronizable if and only if there exists a function $f:\Sigma^* \rightarrow \{-1,1\}$ with $f(\epsilon)$ such that $z$ is an $f$-corner.

\begin{theorem}\label{Thm: Cornering Strategy}
    Let 
    $M = (Q,\Sigma,\delta)$ be a DFA of order $n$,
    and let $f:\Sigma^* \rightarrow \{-1,1\}$ satisfy $f(\epsilon) = 1$. 
    If $z \in Q$ is an $f$-corner of $M$, then $M$ has a $z$-synchronizing word of length at most $\frac 12 d(d+1)(n-1)$, where $d = \max_{q \in Q} \dist^+_f(q,z)$.
\end{theorem}

\begin{proof}
    Let $G = (V,E,\psi)$ be the digraph of $M$, and  identify $V(G)$ with $Q$.
    Suppose that $z$ is an $f$-corner in $M$ and that $d = \max_{q \in Q} \dist^+_f(q,z)$. For each $q \in Q$,
    we let $M_q = (Q,\Sigma,\delta,q,\{z\})$ be a copy of $M$ with an initial active state $q$ and with a single accept state $z$.
    Let $M' = \prod_{q \in Q} M_q$ be the product DFA of all automata $M_q$.
    We aim to construct a word $\pi \in \Sigma^*$
    of length at most $\frac 12 d(d+1)(n-1)$
    that 
    is accepted by $M'$.
    In other words, 
    we aim for
    the action of $\pi$ to map the initial state of each $M_q$ to the accept state $z$.
    One may visualize each automaton $M_q$ as a robot that initially occupies the vertex $q \in V(G)$ and that moves throughout $G$ according to command words from $\Sigma^*$ that it receives. In this setting, we aim to let $\pi$ guide every robot to $z$.

    First, given a fixed vertex $q \in Q$,
    we describe a procedure for constructing a word $\pi_q$
    for which $z(\pi_q) = q(\pi_q) = z$.
    In order to construct the word $\pi_q$, 
    we consider the automata $M_q$ and $M_z$.
   Our goal is to construct a word $\pi_q$
   that maps the initial states of both $M_q$ and $M_z$ to the accept state $z$, so that $\pi_q$ is accepted by both $M_q$ and $M_z$.
    We construct $\pi_q$ over many iterations, as follows. Throughout our iterations, we imagine that words are passed to $M_q$ and $M_z$ which update their active states.

    Iterating through $i = 1, \dots, d$ in order, we define a word $\tau_i$ that we pass to $M_q$ and $M_z$ as follows. If $i$ is odd, then we let $W_i$ be a shortest $f$-positive walk in $G$ from the current active state of $M_q$ to the vertex $z$. If $i$ is even, then we let $W_i$ be a shortest $f$-positive walk in $G$ from the
    active state of $M_z$ to $z$.
    In both cases, we let $\tau_i = \psi(W_i)$. Notice that if $M_q$ and $M_z$ both  have active state  $z$ after iteration $i$, then for all $j>i$, $W_j$ is an empty walk, implying that $\tau_j$ is the empty word $\epsilon$.
    After iterating through $i = 1, \dots, d$ we stop.
    For technical reasons, we also imagine that $i=0$ before we define $\tau_1$, and that we have an iteration corresponding to $i=0$ in which the empty word $\tau_0 = \epsilon$ 
    is passed to $M_q$ and $M_z$. Finally, we write $\pi_q = \tau_0 \tau_1 \dots \tau_d$, so that $\pi_q$ is the entire word that is passed to $M_q$ and $M_z$ over all iterations.

    We claim that 
    both $M_q$ and $M_z$ accept $\pi_q$; in other words,
    after receiving $\pi_q$,
    both $M_q$ and $M_z$ have active state
    $z$. To prove this claim, 
    we show by induction that after iteration $i$, the following holds:
    \begin{itemize}
        \item If $i$ is even, then after $\tau_0\cdots\tau_i$ is passed,
        the shortest $f$-positive walk from the active state of $M_q$ to the vertex $z$ has a
        length of at most $d-i$, and $M_z$ has active state $z$.
        \item If $i$ is odd, then after $\tau_0\cdots\tau_i$ is passed,
        $M_q$ has active state $z$, and a shortest $f$-positive walk from the active state of $M_z$ to  $z$ has a length of at most $d-i$.
    \end{itemize}
    When $i = d$, the statement implies that both $M_q$ and $M_z$ have active state $z$.
    
    To prove the statement, consider a value $0 \leq i \leq d$. 
    For our base case, when $i = 0$, we observe that $M_z$ has active state $z$, and $M_q$ has active state $q$. Then, our theorem's hypothesis implies that $\dist_f^+(q,z) \leq d$, completing the base case. Next, suppose that $i \geq 1$ and that $i$ is odd. By our induction hypothesis, at the beginning of iteration $i$, $M_z$ has active state $z$, and a shortest $f$-positive walk from the active state of
    $M_q$ to $z$ has length at most $d-(i-1)=d-i+1$. Hence, the word $\tau_i$ has length $|W_i| \leq d-i + 1$. 
    Finally, we write $q'$ for the active state of $M_z$ after applying the word $\tau_i$, and we observe that $\dist^+_f(z,q') \leq |W_i| \leq d-i + 1$. Since $z$ is an $f$-corner, $\dist_f^+(q',z) < \dist_f^+(z,q') \leq d-i+1$. Hence, $\dist_f^+(q',z) \leq d-i$. Furthermore, by definition of $\tau_i$, the active state of $M_q$ moves to $z$ after applying the word $\tau_i$. Therefore, 
    at the end of iteration $i$, $M_q$ has active state $z$, and a shortest $f$-positive walk from the active state of $M_z$ to $z$ has length at most $d-i$, as claimed. When $i$ is even, the proof is similar, with $M_q$ and $M_z$ switched. This completes induction.

    We observe that the length of the word $\pi_q = \tau_0 \tau_1 \cdots \tau_d$ passed to $M_q$ and $M_z$ throughout all iterations is at most $d + (d-1) + \dots + 1 = \frac 12 d(d+1)$.
    
    Now, we construct an $z$-synchronizing word $\pi$ in $G$. 
    Recall that for each $q \in Q$, $M_q$ is a copy of $G$ with an initial active state $q$.
    We say that $M_q$ is \emph{synchronized} if $M_q$ and $M_z$ both have active state $z$. Initially, exactly $n-1$ automata $M_q$ are not synchronized.

    We construct $\pi$ as follows. We initialize $\pi = \epsilon$. Then, while there exists some automaton $M_q$ with an active state $p \in Q \setminus \{z\}$, we define a word
    $\pi_p$ of the length at most $\frac{d(d+1)}{2}$ 
    that moves the active states of both $M_q$ and $M_z$ to $z$,
    as described above.
    We pass $\pi_p$ to all automata and update $\pi \leftarrow \pi \pi_p$. We observe that after passing $\pi_p$, the number of synchronized automata $M_q$
    decreases by at least $1$. Since $n-1$ automata are initially not synchronized,
    we can construct an $z$-synchronizing word $\pi$ as a concatenation of at most $n-1$ words $\pi_q$, each of length at most $\frac 12 d(d+1)$. Therefore, we find an $z$-synchronizing word of length at most $\frac 12 d(d+1)(n-1)$, completing the proof.
\end{proof}

We note that the proof of Theorem \ref{Thm: Cornering Strategy} can be modified to produce shorter words of small \emph{rank}, defined as follows. A word $\tau \in \Sigma^*$ has rank $k$ in $M = (Q,\Sigma,\delta)$  if $|\{q(\tau): q \in Q\}|\leq k$. By following the proof of Theorem \ref{Thm: Cornering Strategy}, one can show that for each $k \geq 1$, under the same assumptions, $M$ 
has a word of rank at most $k$ and of length at most $\frac 12 d(d+1)(n-k)$.
In the next subsection, we use Theorem~\ref{Thm: Cornering Strategy} to prove
that all bidirectionally connected difference DFAs in $\mathbb{R}^d$ are synchronizable, 
and that they furthermore often have short synchronizing words.

We have shown that the existence of an $f$-corner is a sufficient condition for a DFA to be synchronizable. 
Now, we show that 
the existence of an $f$-corner is also a necessary condition for a DFA to be synchronizable.

\begin{theorem}\label{Thm: Cornering is necessary}
    Let $M = (Q,\Sigma,\delta)$ be a DFA. For
    each 
    $z \in Q$, $M$ is $z$-synchronizable if and only if there exists a
    function
    $f: \Sigma^* \rightarrow \{-1,1\}$
    satisfying $f(\epsilon) = 1$
    such that $z$ is an $f$-corner.
\end{theorem}

\begin{proof}
    Let $M = (Q,\Sigma,\delta)$ be a DFA and $z \in Q$ a state. If there exists a function $f: \Sigma^* \rightarrow \{-1,1\}$
    satisfying $f(\epsilon) = 1$
    such that $z$ is an $f$-corner, then Theorem~\ref{Thm: Cornering Strategy} implies that $M$ is $z$-synchronizable.
    
    Suppose, on the other hand, 
    that $M$ is $z$-synchronizable.
    Let $\pi$ be a $z$-synchronizing word. Let $f: \Sigma^* \rightarrow \{-1,1\}$ be the function defined so that $f(\tau)=1$ if and only if $\tau = \epsilon$ or $\tau$ is a $z$-synchronizing word. Then, $f(\pi) = 1$. As $\pi$ is $z$-synchronizing, $q(\pi) = z$ for all $q \in Q$; 
    thus, $\dist_f^+(q,z)$ is finite. On the other hand, for each $z$-synchronizing word 
    $\tau$ and each
    $q \in Q\setminus \{z\}$, $z(\tau)=z \neq q$, implying that $\dist_f^+(z,q) = \infty$. Thus, $\dist_f^+(q,z) < \dist_f^+(z,q)$ for all $q\neq z$, implying $z$ is an $f$-corner. This completes the proof.
\end{proof}

\subsection{The cornering strategy for bidirectional connected difference DFAs}
In this section, we prove
that each bidirectional connected difference DFA in $\mathbb{R}^d$ has a sychronizing word of length at most $O(nd^2)$, where $d$
is
the diameter of the DFA (Theorem~\ref{Thm: R^d automata}).
Furthermore,
we prove  in  Theorem~\ref{Thm: R^d Cerny}  that if 
$G$ is the digraph of a bidirectional connected difference DFA of order $n$,
and if each vertex of $G$
has $\Omega(\sqrt{n})$ out-neighbours, then the synchronizing word which results from our construction has the length $O(n^2)$. In doing so,
we show that if every vertex of 
$G$ in $\mathbb{R}^d$ has at least $3\sqrt{\frac{n}{2}}-1$ out-neighbours, then $G$ satisfies \v{C}ern\'y's conjecture.

We begin studying 
bidirectional connected
difference DFAs by
establishing several properties about finite sets of points in $\mathbb{R}^d$.  
Our goal
is to show that for every
bidirectional connected difference DFA $G$ with vertices in $\mathbb{R}^d$, 
there exists 
a well-behaved
function $f:\Sigma^* \rightarrow \{-1,1\}$
for which $G$ has
an $f$-corner.

Given a set $S \subseteq \R^d$, we write $\conv S$ for the convex hull of $S$.
Given 
a convex set $C \subseteq \R^d$,
a point $\mb z \in C$ is an \emph{exposed point} of $C$ if 
there exists a linear functional $C \rightarrow \R$ 
that attains a unique maximum at $\mb z$ \cite[Chapter 8]{Simon}; in other words, there exists $\mb a \in \R^d$ such that $\mb a \dotp \mb x < \mb a \dotp \mb z$ for all $\mb x \in C \setminus \{\mb z\}$.
%
Gr\"unbaum \cite[Chapter 3]{Grunbaum} points out that if $C$
is the convex hull of a finite set $S \subseteq \R^d$, then the exposed points of $C$ are exactly
the
\emph{extreme points} of $C$---that is, the points that cannot be expressed as a proper convex combination of two other points of $C$.
This observation implies that 
if $S \subseteq \R^d$ is a nonempty finite set, then any $\mb z \in S$ for which $\conv (S \setminus \{\mb z\}) \subsetneq \conv S$ is an exposed point in $\conv S$.

\begin{theorem}\label{Thm: R^d automata}
    If $M = (Q,\Sigma,\delta)$ is a bidirectional connected difference DFA of order $n$
    whose digraph has
    diameter $d$,
    then $M$
    has a synchronizing word of length at most 
    $\frac 12 d(d+1)(n-1)$.
\end{theorem}

\begin{proof}
Let $M = (Q,\Sigma,\delta)$
be a directional connected difference DFA of order $n$ with digraph $G= (V,E,\psi)$.
Assume that $Q \subseteq \mathbb R^d$, and identify $V(G)$ and $Q$ for convenience.
By 
our previous discussion,
there exists a point $\mb z \in Q$ and a non-zero vector $\mb a$ such that 
$$
\mb a\dotp (\mb z - \mb x ) > 0
$$
for all $\mb x \in Q \setminus \{\mb z\}$.

We observe that if $P = (\mb p , \dots , \mb q)$ is a path (or more generally a loopless walk)
in $G$, then as $M$ is a difference DFA, the vector sum of the symbols in $\psi(P)$
is
$\mb q - \mb p$.  We let $T$  
be the set of all words $\tau$ for which there exists a directed path $P = (\mb p, \dots, \mb q)$ in $G$ such that $\psi(P) = \tau$ and $\mb a \dotp (\mb q - \mb p) \geq 0$. Then, we let 
$f:\Sigma^* \rightarrow \{-1,1\}$ 
satisfy
$f(\tau) = 1$ if $\tau \in T$ and $f(\tau) = -1$ otherwise. 
We observe that for every directed path $P = (\mb p, \dots, \mb q)$ in $G$,
$f(\psi(P)) = 1 $ if and only if $\mb a \dotp (\mb q - \mb p) \geq 0$.
In particular, $f(\epsilon) = 1$.
Furthermore,
every directed path $P = (\mb p, \dots, \mb z)$ is $f$-positive. 

By Theorem~\ref{Thm: Cornering Strategy} it is sufficient to show there is an $f$-corner in $G$. We claim $\mb z$ is an $f$-corner. For each $\mb q \in Q \setminus \{\mb z\}$, $\dist^+_f(\mb q, \mb z) = \dist(\mb q,\mb z)$,
as every path from $\mb q$ to $\mb z$ is $f$-positive.
Furthermore,
$\dist^+_f(\mb z,\mb q) > \dist(\mb z,\mb q)$ as no
path
from $\mb z$ to $\mb q$ is $f$-positive. 
As $G$ is bidirectional connected, $\dist(\mb q,\mb z)=\dist(\mb z,\mb q)$, so we conclude that 
$$
\dist^+_f(\mb q,\mb z) = \dist(\mb q,\mb z) = \dist(\mb z,\mb q) <\dist^+_f(\mb z,\mb q)
$$ 
so $\dist^+_f(\mb q,\mb z) <\dist^+_f(\mb z,\mb q)$ as desired. 
As
 $\mb q \in Q \setminus \{\mb z\}$ was chosen arbitrarily, we conclude that $\mb z$ is an $f$-corner. This completes the proof.
\end{proof}

For an example of applying Theorem~\ref{Thm: R^d automata} to a fixed DFA, consider the difference DFA $M$ in Figure~\ref{Fig: Diff not aperiodic}. Let ${\mb z} = (0,0)$ represent the bottommost and leftmost state,
and let ${\mb a} = (-1,-1)$. Then, ${\mb a} \cdot {\mb z}  = 0$,
but for any ${\mb q} = (x,y)$ where $x \geq 0$ and $y \geq 0$, ${\mb a} \cdot {\mb q}  \geq 0$ if and only if ${\mb q } = {\mb z}$. 
Hence, for each state $\mb q \in Q \setminus \{\mb z\}$, $\mb a \dotp \mb z > \mb a \dotp \mb q$.

With this example in mind,
we observe that for a general difference DFA $M = (Q,\Sigma,\delta)$, 
whenever ${\mb z}$ is 
an exposed point of $\conv Q$,
we can define $f$ so that $\mb z$ is an $f$-corner of $M$. 
In this case, the convex hull of $Q$ is a polytope, and $\mb z$ is a vertex (or \emph{corner}) of this polytope.
This idea motivates
our name for the cornering strategy.

From Theorem \ref{Thm: R^d automata},
we also see that a 
bidirectional connected
difference DFA whose minimum out-degree
is large must have a short synchronizing word.
\begin{lemma}[{\cite[Theorem.1]{erdHos1989radius}}]\label{Lemma: diameter}
    If $G$ is a simple graph of order $n$ and minimum degree $\delta \geq 2$, then 
    $$
    \diam(G) \leq \frac{3n}{\delta+1}-1.
    $$
\end{lemma}

We note that Lemma~\ref{Lemma: diameter} has been improved in \cite{mukwembi2012note} to  $\diam(G) \leq \frac{3(n-t)}{\delta+1}-1+\frac{3}{\delta+1}$ for $\delta, \diam(G)\geq 5$, where $t$ is the number of distinct terms in the degree sequence of $G$. In some cases this might give a large improvement over Lemma~\ref{Lemma: diameter}. However for general graphs,
the classical result of Erd\H{o}s, Pach, Pollack, and Tuza \cite{erdHos1989radius} given in Lemma~\ref{Lemma: diameter} is sufficient.

\begin{theorem}\label{Thm: R^d Cerny}
    If $M = (Q,\Sigma,\delta)$ is a bidirectional connected difference DFA in $\mathbb{R}^d$ of order $n$,
    and if every vertex $\mb q$ in the digraph of $M$
    satisfies $|N^+(\mb q)| \geq k \geq 2$, then $M$ has a synchronizing word of length at most 
    $$ 
    \frac{9n^3-9n^2}{2(k+1)^2} + \frac{3n-3n^2}{2(k+1)} $$
    Thus, if $n\geq 5$ and $k \geq 3\sqrt{\frac{n}{2}}-1$, then $M$ has a synchronizing word of length strictly less than $(n-1)^2$.
\end{theorem}

\begin{proof}
    Let $G=(V,E,\psi)$ be
    the digraph of $M$,
    and suppose that
    $|N^+(\mb q)| \geq k \geq 2$ 
    for all $\mb q \in V(G)$.
    By Theorem~\ref{Thm: R^d automata},
    $M$ has
    a $\mb z$-synchronizing word of length at most $\frac 12 d(d+1) (n-1) $,
    where $d = \max_{\mb q \in V(G)} \dist^+_f(\mb x, \mb z) \leq \diam(G)$, and $\diam(G)$ is the directed diameter of $G$.

    We note that as $G$ is a bidirectional connected DFA, 
    the directed diameter of $G$ is equal to the diameter of the underlying simple graph of $G$, which we call $H$. Furthermore, we note that for all $\mb q \in V(G)$, $\deg_H(\mb q)-1 \leq |N^{+}_G(\mb q)| \leq \deg_H(\mb q)$. Hence, $\delta(H) \geq k \geq 2$ for all $n \geq 5$. Then, Lemma~\ref{Lemma: diameter} implies that
    \begin{align*}
        d \leq \diam(G) = \diam(H) \leq  \frac{3n}{\delta(H) +1}-1 \leq \frac{3n}{k +1}-1.
    \end{align*}
    Thus, $d(d+1) \leq \frac{9n^2}{(k+1)^2} - \frac{3n}{k+1}$ implying that
    \begin{align*}
        \frac 12 d(d+1)(n-1)  \leq \frac{9n^3-9n^2}{2(k+1)^2} + \frac{3n-3n^2}{2(k+1)},
    \end{align*}
    as desired. Letting $k \geq 3\sqrt{\frac{n}{2}}-1$,
    we note that $\frac{3n}{k +1}-1 \leq \sqrt{2n} - 1$. Thus, for $n\geq 5$,
    $$d(d+1) \leq (2n - 2\sqrt{2n}+1) + (\sqrt{2n}-1) = 2n - \sqrt{2n} < 2(n -1),$$
    which implies that $\frac 12 d(d+1) n < (n-1)^2$. This concludes the proof.
\end{proof}

\section{$f$-ordered DFAs}
\label{sec:f-ordered}

The goal of this section is to 
establish a necessary and
sufficient condition involving partial orders for the existence of a short
synchronizing word in a DFA.
Connections between DFAs and partial orders are often observed
and exploited in automata theory.
For instance, the class of weakly acyclic DFAs can be defined as those DFAs for which the transition relationship induces a partial order (c.f.~\cite{Blondin}).
Furthermore, Trahtman \cite{trahtman2007v}
observed that in an aperiodic DFA, a partial order can be given to states that helps one to construct a synchronizing word efficiently. Using this idea, he proved that an $n$-state aperiodic DFA has a synchronizing word of length at most 
$\frac 12 n(n-1)$.
Ananichev and Volkov
\cite{ananichev2004synchronizing} 
 also
showed that if the states of a synchronizable $n$-state DFA
have a total order that is preserved by the action of every word, then the DFA has a synchronizing word of length at most $n-1$.

The
main result of this section shows that if a DFA
has a partial order on its states satisfying certain properties,
then the DFA has a synchronizing word whose maximum length depends on the properties of the partial order.
We also show that every synchronizable DFA is $f$-ordered for an appropriately defined $f$ and an appropriate partial order.
The rest of the section explores the consequences of these results
by considering the
upper bounds on the lengths of synchronizing words
of different DFA classes
implied by our method.

Let $M = (Q,\Sigma,\delta)$ be a DFA. For a
function $f:\Sigma^* \rightarrow \{-1,1\}$ 
satisfying $f(\epsilon) = 1$,
we say that $M$ is an \emph{$f$-ordered DFA}
if there exists a partial order $\mathcal{P}=(Q,\preceq)$ such that:   
\begin{enumerate}
    \item\label{5.1} For all $f$-positive words $\tau \in \Sigma^*$, $p \preceq q$ implies $p(\tau) \preceq q(\tau)$;
    \item\label{5.2} There exists a universally comparable element $u \in Q$ such that for all $q \in Q$, either $u \preceq q$ or $q \preceq u$;
    \item\label{5.3} 
    $Q$ has a $\mathcal P$-maximal element $m$ such that for all $q \in Q$, there is an $f$-positive word $\tau$ for which $q(\tau) = m$;
\end{enumerate}
An $f$-ordered DFA is pictured in Figure~\ref{fig:ordered-DFA}. 
In this example, $f$ is the function which maps all words to $1$. Notice that the DFA in Figure~\ref{fig:ordered-DFA} 
is not a difference DFA.
In order to prove that every $f$-ordered DFA has a synchronizing word, we need the following lemma from the theory of partially ordered sets.
Given two
partially ordered sets $(Q, \preceq)$ and $(Q', \preceq)$,
a function $g:Q \rightarrow Q'$ is a \emph{homomorphism} if $g(p) \preceq g(q)$ for each pair $p,q \in Q$ satisfying $p \preceq q$.

\begin{figure}\label{Fig:Ordered DFA}
\begin{center}
    \scalebox{0.85}{
        \begin{tikzpicture}[node distance={15mm}, thick, main/.style = {draw, circle}]

\node[draw = none, fill = none] (z) at (8,5) {$u = m$};

\node[main][fill= black] (00) at (0,0) {}; 
\node[main][fill= black] (02) at (0,2) {}; 
\node[main][fill= black] (04) at (0,4) {}; 

\node[main][fill= black] (20) at (2,0) {}; 
\node[main][fill= black] (22) at (2,2) {}; 
\node[main][fill= black] (24) at (2,4) {}; 

\node[main][fill= black] (40) at (4,0) {}; 
\node[main][fill= black] (42) at (4,2) {}; 
\node[main][fill= black] (44) at (4,4) {}; 

\node[main][fill= black] (60) at (6,0) {}; 
\node[main][fill= black] (62) at (6,2) {}; 
\node[main][fill= black] (64) at (6,4) {}; 

\node[main][fill= black] (M) at (7,5) {}; 

\draw [->,blue] (00) to [out=135,in=-135] (02);
\draw [->,ultra thick, dashed, red] (02) -- (00);
\draw [->,blue] (02) to [out=135,in=-135] (04);
\draw [->,ultra thick, dashed, red] (04) -- (02);

\draw [->,blue] (20) to [out=135,in=-135] (22);
\draw [->,ultra thick, dashed, red] (22) -- (20);
\draw [->,blue] (22) to [out=135,in=-135] (24);
\draw [->,ultra thick, dashed, red] (24) -- (22);

\draw [->,blue] (40) to [out=135,in=-135] (42);
\draw [->,ultra thick, dashed, red] (42) -- (40);
\draw [->,blue] (42) to [out=135,in=-135] (44);
\draw [->,ultra thick, dashed, red] (44) -- (42);

\draw [->,blue] (60) to [out=135,in=-135] (62);
\draw [->,ultra thick, dashed, red] (62) -- (60);
\draw [->,blue] (62) to [out=135,in=-135] (64);
\draw [->,ultra thick, dashed, red] (64) -- (62);

\draw [->,blue] (64) to [out=90,in=210] (M);
\draw [->,ultra thick, dashed, red] (M) to  [out=270,in=0] (62);
\draw [->,green!50!black, dotted]  (64) to [out=0,in=240] (M);
\draw [->,cyan, ultra thick, densely dotted] (M) to [out=180,in=90] (44);

\draw [->,green!50!black, dotted] (00) -- (20);
\draw [->,cyan, ultra thick, densely dotted] (20) to [out=135,in=45](00);
\draw [->,green!50!black, dotted] (20) -- (40);
\draw [->,cyan, ultra thick, densely dotted] (40) to [out=135,in=45](20);
\draw [->,green!50!black, dotted] (40) -- (60);
\draw [->,cyan, ultra thick, densely dotted] (60) to [out=135,in=45](40);

\draw [->,green!50!black, dotted] (02) -- (22);
\draw [->,cyan, ultra thick, densely dotted] (22) to [out=135,in=45](02);
\draw [->,green!50!black, dotted] (22) -- (42);
\draw [->,cyan, ultra thick, densely dotted] (42) to [out=135,in=45](22);
\draw [->,green!50!black, dotted] (42) -- (62);
\draw [->,cyan, ultra thick, densely dotted] (62) to [out=135,in=45](42);

\draw [->,green!50!black, dotted] (04) -- (24);
\draw [->,cyan, ultra thick, densely dotted] (24) to [out=135,in=45](04);
\draw [->,green!50!black, dotted] (24) -- (44);
\draw [->,cyan, ultra thick, densely dotted] (44) to [out=135,in=45](24);
\draw [->,green!50!black, dotted] (44) -- (64);
\draw [->,cyan, ultra thick, densely dotted] (64) to [out=135,in=45](44);

        \end{tikzpicture}
    }
\end{center}
\caption{The DFA in the figure is an example of a 
$f$-ordered DFA whose vertex set is a subset of $\mathbb Z^2$, equipped with the partial order $(x_1, y_1) \preceq (x_2,y_2)$ if and only if $x_1 \leq x_2$ and $y_1 \leq y_2$. 
Each edge color corresponds to an alphabet symbol, and loops are not drawn in the figure. 
}

\label{fig:ordered-DFA}
\end{figure}

\begin{lemma}
\label{lem:surj}
    Let $(Q, \preceq)$ and $(Q', \preceq)$ be partially ordered sets, and let $g: Q\rightarrow Q'$ be a surjective homomorphism.
    Let $M$ be the set of minimal elements of $(Q, \preceq)$, and let $M'$ be the set of minimal elements of $(Q', \preceq)$.
    Then, $M' \subseteq g(M)$.
\end{lemma}
\begin{proof}
    Suppose that the lemma is false, and let $q'$ be a minimal element of $(Q', \preceq)$
    that does not belong to $g(M)$.
    Let $q$ be in the preimage of $q'$, and note that $q$ is not minimal in $(Q, \preceq)$.
    Therefore, there exists an element $p \in Q$ for which $p \prec q$.
    Since $g(p) \neq q'$, and since $g$ is a homomorphism, we have $g(p) \prec g(q) = q'$, contradicting the assumption that $q'$ is minimal.
\end{proof}

Our proof employs a similar approach to that of Theorem~\ref{Thm: R^d automata}.
Specifically, our strategy shows that repeatedly moving active states of an $f$-ordered DFA to the maximal element $m$ in an appropriate way must produce a synchronizing word; in this sense, the maximal element $m$ functions similarly to the corner state in the cornering strategy.
However, we are often able to take advantage of the partially ordered structure to make our synchronizing word shorter.

\begin{theorem}\label{Thm: Partial Ordered DFA}
    Let $M = (Q,\Sigma,\delta)$
    be an $f$-ordered DFA with a partial order $\mathcal{P} = (Q,\preceq)$, 
    and let $u,m \in Q$ be as specified in (\ref{5.2}) and (\ref{5.3}).
    Suppose that:
    \begin{itemize}
        \item $\dist_f^+(q,m) \leq d$ for every $q \in Q$;
        \item $\mathcal P$ has at most $k$ minimal elements.
    \end{itemize}
    Then, $M$ has a synchronizing word
    of length at most $dk$.
\end{theorem}

\begin{proof}
For each $q \in Q$, let $\tau_q$ be a shortest $f$-positive word for which $q(\tau_q) = m$; note that $\tau_q$ exists by (\ref{5.3}).
 Note also that each word
 $\tau_q$ has length at most $d$.
 Consider also a pair $p,q \in Q$ such that $p \prec q$.
 As $\tau_p$ is $f$-positive, we observe from
 (\ref{5.1}) that $p(\tau_p) = m\preceq q(\tau_p)$. As $p(\tau_p) = m$ is maximal,
 this implies that $p(\tau_p) = m = q(\tau_p)$. 

For each $q \in Q$, we 
define $M_q = (Q,\Sigma,\delta,q,\{m\})$, so that $M_q$ is a copy of $M$
with an initial active state $q$ and a single accept state $m$.
We let $M' = \prod_{q \in Q} M_q$, so that $M'$ accepts a word $\tau$ if and only if $\tau$ is an $m$-synchronizing word.
We aim to construct a word $\pi$ accepted by $M'$.
We construct $\pi$ in iterations, and during these iterations, we pass words to the automata $M_q$, updating their active states.

We write $\pi_0 = \epsilon$ for the empty word.
For each $i \geq 0$,
whenever we define a word $\pi_i$, we assume that the word $\pi_i$
is passed to all automata $M_q$, updating their active states.
For each $i \geq 0$ for which $\pi_i$ is defined, we let $Q_i$ denote the set of active states of all automata $M_q$ after $\pi_0 \cdots \pi_i$ is passed.
Initially, $Q_0 = Q$.

We proceed in iterations. For $i \geq 1$, we 
proceed as follows.
If $|Q_{i-1}| = 1$, then 
then the procedure halts without executing iteration $i$.
Otherwise, iteration $i$ is executed as follows. 
We assume that the words $\pi_0,\pi_1,\dots, \pi_{i-1}$ have already been defined and passed to our automata in that order. We proceed in cases based on the active states of our automata, as follows.

\vspace{0.25cm}
\underline{Case 1:} If there exists an automaton $M_q$ with active state $p \prec m$, then:

\vspace{0.25cm}
Choose $M_q$ so that the active state $p \prec m$ of $M_q$ is minimal in the partially ordered set $(Q_{i-1}, \preceq)$.
Then,
pass the word $\pi_i = \tau_p$ to all automata. 
As the action of $\pi_i$ preserves the partial order between states,
$p(\pi_i) = m = m(\pi_i)$.
Then, proceed to iteration $i+1$, in which the process halts or follows Case 1 or 2.

\vspace{0.25cm}
\underline{Case 2:} If  no automaton has an active state $p \prec m$, then:

\vspace{0.25cm}
Let $\pi_i = \tau_u$. 
We observe that $u \prec p$ for each active state $p \in Q_{i-1}$, as otherwise $p \preceq u \preceq m$, a contradiction. Therefore, for every active state $p \in Q_{i-1}$, $ u(\tau_u) = m(\tau_u)  = p(\tau_u) = m$.
Then, $Q_i = \{m\}$,
and our procedure halts before iteration $i+1$.

\vspace{0.25cm}
\begin{claim}
    For each iteration $i$ executed by our algorithm,
    the number of minimal elements in 
$(Q_i \setminus \{m\}, \preceq)$ is less than the number of minimal elements in $(Q_{i-1} \setminus \{m\}, \preceq)$.
\end{claim}
\begin{proof}
Consider an iteration $i$ of our algorithm.
Note that the number of minimal elements in  $(Q_{i-1} \setminus \{m\}, \preceq)$ is at least $1$, or else our algorithm would halt before executing iteration $i$.
If iteration $i$ follows Case 2, then 
$Q_i = \{m\}$, and hence
$(Q_i \setminus \{m\}, \preceq)$
has no minimal element, and the claim holds.
Otherwise, suppose that Case 1 is followed.
As $\pi_i$ is an $f$-positive word, the action $g:Q_{i-1} \rightarrow Q_i$ of $\pi_i$ on the set $Q_{i-1}$ is a homomorphism with respect to the partial order $\preceq$. Furthermore, clearly $g$ is surjective.
Therefore, 
by Lemma \ref{lem:surj},
each minimal element of $(Q_i, \preceq)$ is the image of a minimal element of $(Q_{i-1}, \preceq)$ under $g$.
Furthermore, as observed above, $m(\pi_i) = m$.
Hence,
each minimal element of $(Q_i \setminus \{m\}, \preceq)$ is the image of a minimal element of $(Q_{i-1} \setminus \{m\}, \preceq)$ under $g$.
Finally, $g$ maps some minimal element $p$ of 
$(Q_{i-1} \setminus \{m\}, \preceq)$ to $m$.
Therefore, $p(\pi_i)$ is not a minimal element in 
$(Q_i \setminus \{m\}, \preceq)$, and 
our claim holds.
\end{proof}

By our claim,
we see that our algorithm terminates after executing at most $k$ iterations.
Indeed, 
as $(Q_0 \setminus \{m\}, \preceq)$ has at most $k$ minimal elements, it follows that
after some number $\ell \leq k$ iterations, $(Q_{\ell} \setminus \{m\}, \preceq)$ has no minimal element, implying $Q_{\ell} = \{m\}$.
In other words, 
for some $\ell \leq k$, $\pi_0 \cdots \pi_{\ell}$ is an $m$-synchronizing word.
As the word $\pi_i$ produced on each iteration $i$ has length at most $d$,
our algorithm produces a synchronizing word of length at most $dk$.
This completes the proof.
\end{proof}

The following theorem shows that in fact every synchronizable DFA is $f$-ordered for an appropriately defined function $f$ and an appropriate partial order.

\begin{theorem}
\label{thm:f-char}
    Let $M = (Q,\Sigma,\delta)$ be a DFA. Then, 
    $M$ is $f$-ordered for some
    function $f: \Sigma^* \rightarrow \{-1,1\}$ satisfying $f(\epsilon) = 1$
    if and only if $M$ is synchronizable.
\end{theorem}
\begin{proof}
    If $M$ is $f$-ordered for some
    function $f: \Sigma^* \rightarrow \{-1,1\}$ satisfying $f(\epsilon) = 1$,
    then
    Theorem \ref{Thm: Partial Ordered DFA} implies that $M$ is synchronizable. 
    On the other hand, suppose that $M$ has a $z$-synchronizing word for some $z \in Q$.
    We define a partial order $\mathcal P = (Q,\preceq)$ 
    on $Q$ by letting $q \preceq z$ for all $q \in Q$ and by letting each pair $p,q \in Q \setminus \{z\}$ be incomparable.
    In other words, the Hasse diagram of $\mathcal P$ is a star with a center $z$ that lies above
    $|Q| - 1$ leaves.
    We let $f(\epsilon) = 1$, and for
    each
    nonempty $\tau \in \Sigma^*$, we let $f(\tau) = 1$ if and only if $\tau$ is a $z$-synchronizing word.
    Then, we check that the definition of an $f$-ordered DFA holds for $M$ when letting $z$ play the role of both $m$ and $u$.

    If $p \preceq q$ and 
    $\tau$ is a nonempty $f$-positive word, then
    $p(\tau) = q(\tau) = z$.
    Furthermore, clearly $p(\epsilon) \preceq q(\epsilon)$. Therefore,
    (\ref{5.1}) holds. Since $q \preceq z$ for all $q \in Q$, (\ref{5.2}) holds. Since $z$ is $\mathcal P$-maximal, and since each state $q \in Q$ is mapped to $z$ by a $z$-synchronizing word, (\ref{5.3}) holds.
    Therefore, $M$ is an $f$-ordered DFA.
\end{proof}

We observe that
Theorem \ref{Thm: Partial Ordered DFA} has several corollaries related to existing DFA classes.
Our first corollary is related to
DFAs with a partial order on their states that is preserved by the action of every word.
Given a partially ordered set $\mathcal P = (X, \preceq)$ we say that an element $m \in X$ is \emph{maximum} if $x \preceq m$ for all $x \in X$.

\begin{corollary}\label{Coro: max element f-everything}
    Let $M = (Q,\Sigma,\delta)$ be a DFA. Suppose there exists a partial order $\mathcal{P}=(Q,\preceq)$ such that  
    \begin{enumerate}
        \item for all words $\tau \in \Sigma^*$, $p \preceq q$ implies $p(\tau) \preceq q(\tau)$, and
        \item $\mathcal P$ has a universally reachable maximum element $m$.
    \end{enumerate}
    If $\dist(q,m) \leq d$ for each $q \in Q$
    and $\mathcal P$ has at most $k$ minimal elements,
    then $Q$
    has a synchronizing word of length at most $dk$.
\end{corollary}

\begin{proof}
    Define a function $f: \Sigma^* \rightarrow \{-1,1\}$ that maps all words to $1$.
    We observe that $M$ is $f$-ordered, and that the maximum element of $\mathcal P$ plays the role of both $m$ and $u$ in the definition of the $f$-ordering.
    Thus, by Theorem \ref{Thm: Partial Ordered DFA},
    $M$ has a 
    synchronizing word of length $dk$.
\end{proof}

We note that Corollary \ref{Coro: max element f-everything}
implies a special case of the following result of Ananichev and Volkov \cite{ananichev2004synchronizing}:
\begin{theorem}[\cite{ananichev2004synchronizing}]
    Let $M = (Q, \Sigma,\delta)$ be a DFA, and suppose that $Q$ has a total order that is preserved by the action of every word of $\Sigma^*$.
    If $M$ has a word of rank $r$, then some rank-$r$ word of $M$ has length at most $n-r$. In particular, if $M$ is synchronizable, then some synchronizing word of $M$ has length at most $n-1$.
\end{theorem}
Using Corollary \ref{Coro: max element f-everything},
if such a 
 DFA $M = (Q,\Sigma,\delta)$ with a total order has a universally reachable maximum element $m \in Q$,
then we may define $f:\Sigma^* \rightarrow \{-1,1\}$ so that every word $\tau \in \Sigma^*$ is $f$-positive.
Then, it is easy to see that $\dist(q,m) \leq n-1$ for each $q \in Q$. Since $\mathcal P$ is a total order, $\mathcal P$ has only one minimal element, which implies that $M$ has a synchronizing word of length at most $n-1$.

It is also interesting to note that Theorem \ref{Thm: Partial Ordered DFA}
nearly implies the exact length of a shortest synchronizing word
of \v{C}ern\'y's
DFA
 \cite{vcerny1964poznamka} that is conjectured to be extremal for \v{C}ern\'y's conjecture, which defined as follows.
Let $M = (Q,\Sigma,\delta)$, $Q  = \mathbb Z_n$, and $\Sigma = \{\texttt a, \texttt b\}$.
Let $Q$ have a total order $0 \leq 1 \leq \cdots \leq n-1$.
Let $\delta(0,\texttt a) = n-1$, and let $\delta(t,\texttt a) = t$ for each $t \in \mathbb Z_n \setminus \{0\}$. Let $\delta(t,\texttt b) = t-1$ for each $t \in \mathbb Z_n$.
Let $f$ be defined so that the $f$-positive words are 
exactly the words generated by concatenations of $\texttt b^{n-1} \texttt a$.
Then, it is straightforward to see that $M$ is an $f$-ordered DFA with a maximal element $n-1$. 
Furthermore, for each state $q \in Q$, at most $n-1$ applications of $\texttt b^{n-1} \texttt a$ 
map $q$ to $n-1$.
Therefore, applying Theorem \ref{Thm: Partial Ordered DFA} with
$m = u = n-1$, $d = n(n-1)$, and $k = 1$, we see that $M$ has a synchronizing word of length at most $n(n-1)$, which is close to the correct value of $(n-1)^2$.

Finally, we 
note that Theorem~\ref{Thm: Partial Ordered DFA} implies
an upper bound on the length of synchronizing words in aperiodic DFAs that arises naturally from an existing argument of Trahtman \cite[Lemma 4]{trahtman2007v}.

\begin{corollary}\label{Corollary: Aperiodic short words}
Let $M = (Q, \Sigma, \delta)$
    be an $n$-state aperiodic DFA, and let $\ell$ be an integer such that for every $q \in Q$,
    $q(\tau^{\ell}) = q(\tau^{\ell+1})$.
    If $M$ is synchronizable, then $M$ has a synchronizing word of length $(n-1)(\ell+1)d$ where $d = \min_{z\in Q}\max_{q \in Q} \dist(q,z)$.
\end{corollary}

\begin{proof}    
    As $M$ is synchronizable, $M$ has a universally reachable state.
    We choose a state $z \in Q$ to minimize $d = \max_{q \in Q} \dist(q,z)$.
    Note that since $z$ can (and must) be chosen to be universally reachable, $d$ is finite.
    For each $q \in Q$, let $\tau_q$ be a word of length at most $d$ for which $q(\tau_q) = z$.
    
    Let $f: \Sigma^* \rightarrow \{-1,1\}$ be the function that satisfies $f(\tau)=1$ if and only if $z(\tau)=z$. We define the partial order $\mathcal{P}=(V,\preceq)$ as follows; 
    $q \preceq z$ for each $q \in Q$, and $q$ and $p$ are incomparable for each distinct pair $q,p \in Q \setminus \{z\}$.
    Note that $\mathcal P$ has $n-1$ minimal elements.
    
    We claim that $\mathcal P$ is an $f$-ordering of $G$.
    To show (\ref{5.1}), suppose $p,q \in Q$ and $p \preceq q$, and let $\tau \in \Sigma^*$ be an $f$-positive word. If $p=q$, then $p(\tau) = q(\tau)$. Otherwise, $q = z$, so that $p(\tau) \preceq z = q(\tau)$.
    To show (\ref{5.2}), let $u = z$, and observe that $z$ is comparable to every element of $Q$. To show (\ref{5.3}), let $m = z$.
    Then, note that for each $q \in Q$, $q(\tau_q^{\ell}) = q(\tau_q^{\ell+1}) = z(\tau_q^{\ell}) = p$ for some state $p \in Q$, so that $\pi:=\tau_q^{\ell} \tau_p$ is an $f$-positive word of length at most $(\ell+1)d$
    satisfying $q(\pi) = z$.

    To finish the proof, we apply Theorem \ref{Thm: Partial Ordered DFA} with $u = m = z$, $k=n-1$, and with $(\ell+1) d $ in the place of $d$.
\end{proof}

It is also natural to ask if our assumption in Theorem~\ref{Thm: Partial Ordered DFA} that there exists a universally comparable state $u$ can be dropped.
That is, does there exist a DFA $G$ admitting a non-trivial partial order $\mathcal{P}$ with no universally comparable state, where $\mathcal{P}$ is preserved by the action of every word, but $G$ is not synchronizable? The answer is yes. 

For example,
consider the following DFA 
$M = (Q,\Sigma,\delta)$
for which $Q = \mathbb Z_n \times \mathbb Z_2$.
We let $\Sigma = \{\texttt{up},\texttt{down},\texttt{turn}\}$.
For each $(p,q) \in \mathbb Z_n \times \mathbb Z_2$,
we let $\delta((p,q), \texttt{up}) = (p,1)$, and we let 
$\delta((p,q), \texttt{down}) = (p,0)$.
Furthermore, we let $\delta((p,q),\texttt{turn}) = (p+1,q)$.
We also define a partial order $\mathcal P = (Q, \preceq)$ 
so that reflexivity holds, $(p,0) \preceq (p,1)$ for all $k \in \mathbb Z_n$, and so that every other pair of elements in $Q$ is incomparable. 
We let $f(\tau) = 1$ for every $\tau \in \Sigma^*$, and we observe the action of every word of $\Sigma^*$ preserves the partial order $\preceq$.
Then, $M$ is not synchronizable,
as for all distinct $p,p' \in \mathbb Z_n$, $q \in \mathbb Z_2$,
and words $\tau \in \{\texttt{up},\texttt{down},\texttt{turn}\}^*$, $(p,q)(\tau) = (p+x,q') \neq (p'+x,q') = (p',q)(\tau)$, where $q' \in \mathbb Z_2$, and $x$ is the number of times \texttt{turn} appears in $\tau$.

\section{Synchronizing Product DFAs}
\label{sec:product}
In this section we consider the problem of synchronizing a DFA $M_1$ to a prescribed state $z_1$ and a second disjoint DFA $M_2$ to a prescribed
state $z_2$, using a common
word $\tau$. 
We do not assume that $M_1$ and $M_2$ are related beyond having the same alphabet $\Sigma$. This is a natural problem to consider in certain applications. For example, if we take $M_1$ and $M_2$
to be distinct difference DFAs in $\mathbb{R}^2$ with a common state $z$, 
and if we let $z_1 = z_2 = z$, then this
setting
models the problem of 
using the same commands to guide two delivery robots through different parts of a warehouse so that both robots simultaneously arrive at the front door to drop off their respective packages. 
The assumption of needing a common set of commands for all machines is reasonable when using simple machines that cannot receive individualized commands.

One natural model for this situation is a product DFA.
If $M_1$ and $M_2$ are DFAs on the same alphabet
with respective digraphs $G_1 = (V_1,E_1,\psi_1)$ and $G_2 = (V_2,E_2,\psi_2)$, then the definition of a product DFA in Section \ref{sec:prelim}
implies that $M_1 \times M_2$ is represented by the following edge-labeled digraph $G = (V,E,\psi)$:
\begin{enumerate}
    \item $V(G) = V_1 \times V_2$;
    \item $((a,b),(u,v)) \in E(G)$ if and only if $(a,u)\in E_1$, $(b,v) \in E_2$, and $\psi_1(a,u) = \psi_2(b,v)$;
    \item   $\psi((a,b),(u,v)) = \psi_1(a,u) = \psi_2(b,v)$ for all $((a,b),(u,v)) \in E(G)$.
\end{enumerate}

For an example about how this operation is of natural interest, suppose that we have a difference DFA $M_1$ 
that models a location of a simple machine
and a DFA $M_2$ which models the current state of the same machine. As the machine is simple, suppose that the commands it receives to change location and state are the same. Constructing such machines may be cheaper, or in some contexts such as for microscopic machines, necessary. As the controller, we are interested in the state and location of the machine. Thus, the DFA $M_1 \times M_2$ is required, as we want the machine to be in the correct location with the correct state to accomplish its task.

The question now is, 
given $M_1 = (Q_1, \Sigma, \delta_1)$ and $M_2 = (Q_2, \Sigma, \delta_2)$,
for which $z_1 \in Q_1$ and $z_2 \in Q_2$ does there exist a synchronizing word which results in an active state of $(z_1,z_2)$ in $M_1 \times M_2$? We note that if $M_1 \times M_2$ is strongly connected and synchronizable, then any active state of $M_1 \times M_2$ is achievable by a synchronizing word. 
That is,
to reach the desired active state $a$, simply apply a synchronizing word $\tau$ to reach some
state $b$, then apply a word $\pi$ for which $b(\pi) = a$.
However, even if $M_1$ and $M_2$ are strongly connected,
it may not hold that 
$M_1 \times M_2$ is strongly connected.  See Figure~\ref{fig:DFA-sum} 
for an example of two strongly connected DFAs $M_1$ and $M_2$ whose product is not strongly connected. 
Hence, such a word $\pi$ need not exist for every $(z_1, z_2) \in Q_1 \times Q_2$. 
Hence,
it is worth asking for which states $(z_1, z_2) \in Q_1 \times Q_2$ the automaton $M_1 \times M_2$ has a $(z_1, z_2)$-synchronizing word.

\begin{figure}\label{Fig:GxH}
\begin{center}
    \scalebox{0.85}{
        \begin{tikzpicture}[node distance={15mm}, thick, main/.style = {draw, circle}] 

\node[main][fill= black] (v) at (1,1) {};
    \node[fill=none] at (1.5,1) (nodes) {$q$};
\node[main][fill= black] (u) at (-1.5,1) {};
    \node[fill=none] at (-2,1) (nodes) {$p$};
\node[fill=none] at (-3,1) (nodes) {$M_1:$};

\node[main][fill= black] (z) at (4,0) {};
    \node[fill=none] at (4,-0.5) (nodes) {$z$};
\node[main][fill= black] (y) at (5,1) {};
    \node[fill=none] at (5.5,1) (nodes) {$y$};
\node[main][fill= black] (x) at (4,2) {}; 
    \node[fill=none] at (4,2.5) (nodes) {$x$};
\node[fill=none] at (2.5,1) (nodes) {$M_2:$};

\draw [->,black] (v) to (u);
\draw [->,black] (u) to [out=45,in=135] (v);
\draw [->, ultra thick, dashed, red] (u)  to [out=-45,in=-135]  (v);

\draw [->,black] (y) to (z);
\draw [->,black] (z) to (x);
\draw [->, ultra thick, dashed, red] (x) to (y);

\node[fill=none] at (-2.5,-3) (nodes) {$M_1 \times M_2:$};

\node[main][fill= black] (6) at (4,-4) {};
\node[main][fill= black] (5) at (5,-3) {};
\node[main][fill= black] (4) at (4,-2) {};

\node[main][fill= black] (3) at (0,-4) {};
\node[main][fill= black] (2) at (1,-3) {};
\node[main][fill= black] (1) at (0,-2) {}; 

\draw [->,black] (1) to (4);
\draw [->, ultra thick, dashed, red] (1) to (5);

\draw [->,black] (2) to (6);
\draw [->, ultra thick, dashed, red] (2) to (5);

\draw [->,black] (3) to (4);
\draw [->, ultra thick, dashed, red] (3) to (6);

\draw [->,black] (4) to [out=135,in=45] (1);
\draw [->, ultra thick, dashed, red] (4) to (5);

\draw [->,black] (5) to [out=-90,in=-45] (3);

\draw [->,black] (6) to [out=-130,in=-90] (1);
        \end{tikzpicture}
    }
\end{center}
\caption{
An example of a product 
of DFAs with no loops drawn. In the drawing of $M_1 \times M_2$, the three vertices on the left have the form $(p,\cdot)$ while the three vertices on the right have the form $(q,\cdot)$.}
\label{fig:DFA-sum}
\end{figure}

We are able to prove
that if DFAs $M_1$ and $M_2$ have $f$-corners $z_1$ and $z_2$, respectively, then $M_1 \times M_2$ is $(z_1,z_2)$-synchronizable. It is unclear if $(z_1,z_2)$ is always an $f$-corner in $M_1 \times M_2$. Thus, our proof employs a slightly more complicated approach than the cornering strategy; 
however the fundamentals of the proof are the same.

\begin{theorem}\label{Thm: Two Corner Disjoint DFAs}
    Let $M_1 = (Q_1,\Sigma,\delta_1)$
    and $M_2 = (Q_2, \Sigma, \delta_2)$
    be DFAs with a common alphabet $\Sigma$ with digraphs $G$ and $H$ respectively, and let $f:\Sigma^* \rightarrow \{-1,1\}$ satisfy $f(\epsilon) = 1$.
    If $z_1$ is an $f$-corner in $M_1$ and $z_2$ is an $f$-corner in $M_2$, then $M_1 \times M_2$ has a
    $(z_1,z_2)$-synchronizing word
    of length at most 
    \[
     \frac{(n_1+n_2+1)d^2 + (n_1+n_2)}{2},
    \]
    where $|Q_1|= n_1$, $|Q_2| = n_2$ and $d = \min \{ \max_{q\in Q_1}\dist^+_{G,f}(q,z_1), \max_{p\in Q_2}\dist^+_{H,f}(p,z_2)\}$.
\end{theorem}

\begin{proof}
     We let $G = (V_1,E_1,\psi_1)$ be the digraph representation of $M_1$, and we let $H = (V_2, E_2,\psi_2)$ be the digraph representation of $M_2$.
     Without loss of generality, assume $d = \max_{q\in Q_1}\dist^+_{G,f}(q,z_1)$.
     For convenience, we identify $V_1$ with $Q_1$ and $V_2$ with $Q_2$.
     We note that by Theorem~\ref{Thm: Cornering Strategy},
     $M_1$ is $z_1$-synchronizable and $M_2$ is $z_2$-synchronizable.

     Let $\tau_1$ be a $z_1$-synchronizing word for $M_1$,
     and let $\tau_2$ be a $z_2$-synchronizing word for $M_2$. 
     Then,
     for 
     each state $(q_1, q_2) \in Q_1 \times Q_2$,
     we have $(q_1,q_2)(\tau_1\tau_2) = (q,z_2)$, where $q = z_1(\tau_2)$. It remains to show that there exists a 
    word $\pi \in \Sigma^*$
    for which 
      $(q,z_2)(\pi) = (z_1,z_2)$.

     To construct our word $\pi,$ we iterate $i = 0,1, \dots,  \dist_{G,f}^+(q,z_1)$, and on each iteration, we define a word $\pi_i$.
    Ultimately, our word $\pi$ will be the concatenation of all words $\pi_i$.
    
    For each $i$, we write $(x_i, y_i) = (q,z_2)(\pi_0 \cdots \pi_i)$.
     We also let $D_i = \dist^+_{G,f}(x_i,z_1)+\dist^+_{H,f}(y_i,z_2)$. 
     Note that if $D_i = 0$, then $(q,z_2)(\pi_0 \cdots \pi_i) = (z_1,z_2)$.

    For iteration $i = 0$, we define $\pi_0 = \epsilon$.
    Thus,
    $(x_0,y_0) = (q,z_2)$, and 
    $D_0 = \dist_{G,f}^+(q,z_1)$.

    Now, suppose that $i \geq 1$.
     At the beginning of iteration $i$,
     we consider the automaton $M_1 \times M_2$ with active state $(x_{i-1},y_{i-1})$.
     If $i$ is odd,
     let $W_i$ be a shortest $f$-positive walk in $G$ from $x_{i-1}$ to $z_1$, and
     if $i$ is even,
     let $W_i$ be a shortest $f$-positive walk in $H$ from $y_{i-1}$ to $z_2$. 
     Let $\pi_i = \psi_j(W_i)$,
     where $j \in \{1,2\}$ has the same parity as $i$.
     Then, $x_i = x_{i-1}(\pi_i)$ and $y_i = y_{i-1}(\pi_i)$. This completes iteration $i$.

    Observe that if $i$ is odd then $x_i = z_1$,
    and if $i$ is even then $y_i = z_2$. So, if $i\geq 1$ is odd, then
     $$
     D_{i} = \dist_{H,f}^+(y_i,z_2) < \dist_{H,f}^+(z_2,y_i) \leq |\pi_i| = D_{i-1},
     $$
    as $z_2$ is an $f$-corner and $y_i = y_{i-1}(\pi_i) = z_2(\pi_i)$ for an $f$-positive word $\pi_i$. Similarly if $i \geq 1$ is even, then
     $$
     D_{i} = \dist_{G,f}^+(x_i,z_1) < \dist_{G,f}^+(z_1,x_i) \leq |\pi_i| = D_{i-1},
     $$
    as $z_1$ is an $f$-corner and $x_i = x_{i-1}(\pi_i) = z_1(\pi_i)$ for an $f$-positive word $\pi_i$.

    Thus, for all $1 \leq i \leq \dist^+_{G,f}(q,z_1)$, $D_{i} < D_{i-1}$,
    implying that when $k = \dist^+_{G,f}(q,z_1)$, $D_k = 0$.
    Letting $\pi = \pi_0 \cdots \pi_k$,
    we thus see that $(q,z_2)(\pi) = (z_1,z_2)$.
    Therefore, 
    $\tau_1 \tau_2 \pi$ is a $(z_1,z_2)$-synchronizing word. 

    Now let $d = \min \{ \max_{q\in V(G)}\dist^+_{G,f}(q,z_1), \max_{p\in V(H)}\dist^+_{H,f}(p,z_2)\}$.
    Observe that $D_0 \leq \dist^+_{G,f}(q,z_1) \leq d$.
    Hence, $|\pi| \leq d^2$.
    Moreover, if $M_1$ has order $n_1$ and $M_2$ has order $n_2$, 
    then Theorem~\ref{Thm: Cornering Strategy} implies
    \[
    |\tau_1| \leq \frac 12 d(d+1) n_1
    \]
    and 
    \[
    |\tau_2| \leq \frac 12 d(d+1) n_2.
    \]
    Therefore, $\tau_1 \tau_2 \pi$ has length at most 
    \begin{align*}
        \frac{(n_1+n_2+1)d^2 + (n_1+n_2)}{2}.
    \end{align*}
    This completes the proof.
\end{proof}

Notice that this bound implies a stronger result than \v{C}ern\'y's conjecture
for some product DFAs.
For example,
Theorem~\ref{Thm: Two Corner Disjoint DFAs} implies that if $M_1$ and $M_2$ are difference DFAs, both with an $f$-corner where $f$ is the function described in Theorem~\ref{Thm: R^d automata}, and both of order $\sqrt{n}$,
then $M_1\times M_2$ has order $n$ and 
admits synchronizing word of length $(1+o(1))n^{3/2}$.
This is much less than the $(n-1)^2$ bound predicted by \v{C}ern\'y's conjecture.

\section{Future Work}

While we have resolved several questions regarding the synchronizablility of DFAs,
many more questions
remain open. Chief among these is \v{C}ern\'y's conjecture;
however we will focus our attention here on more manageable, or at least more novel open problems.

The first of these problems is whether or not the bound given in Theorem~\ref{Thm: Cornering Strategy} for the length of a shortest synchronizing word is tight.
That is, do there exist DFAs whose shortest synchronizing words have length approximately $\frac 12 d(d+1)(n-1)$ where $d = \max_{q \in Q} \dist_f^+(q,z)$ for an $f$-corner $z$? Moreover, if the answer is yes, then does this remain true if we also require that $d \geq N$ for some large integer $N$?

\begin{conjecture}\label{Conj: Tight bound}
    Let $\Sigma$ be an alphabet, let $f:\Sigma^* \rightarrow \{-1,1\}$ satisfy $f(\epsilon) = 1$, and let $d \in \mathbb{N}$.
    There exists a constant $c$ independent of $d$ and DFA a 
    $M = (Q,\Sigma,\delta)$ of some order $n$ with an
    $f$-corner $x$,
    where $d = \max_{q \in Q} \dist^+_f(q,x)$ and each synchronizing word for $M$ has length at least $cnd^2$.
\end{conjecture}

Alternatively we can ask whether the assumptions of bidirectionality and of each vertex having a large out-neighbourhood in Theorem~\ref{Thm: R^d Cerny} are required.

\begin{conjecture}\label{Conj: Diff DFA satisfies Cerny}
    If $M = (Q,\Sigma,\delta)$ is an $n$-state difference DFA in $\mathbb{R}^d$ on $n$ vertices with a universally reachable state, then $M$ has a synchronizing word
    of length at most $(n-1)^2$.
\end{conjecture}

To conclude,
we ask questions regarding the product of two or more DFAs. Out of all of the questions we give here,
this question 
seems most applicable, especially as one of our primary motivations for considering the product is simple machines, which in modern contexts are often deployed at scale.

\begin{conjecture}\label{Conj: Disjoint Corner DFAs }
    Let $k \geq 3$, let $\Sigma$ be an alphabet, and let $f:\Sigma^* \rightarrow \{-1,1\}$ satisfy $f(\epsilon) = 1$. 
    If $M_1,\dots, M_k$ are 
    distinct DFAs over $\Sigma$
    and for all $i \in [k]$ there exists an $f$-corner $z_i$ in $M_i$, then $M = M_1 \times \dots \times M_k$ is $(z_1,\dots, z_k)$-synchronizable. 
\end{conjecture}

Another natural question to ask is the following.

\begin{question}\label{Conj: Everywhere sychronizable}
    For which DFAs $M_1$ and $M_2$ over the same alphabet is $M_1 \times M_2$ $(z_1,z_2)$-synchronizable for all states $z_1$ in $M_1$ and $z_2$ in $M_2$?
\end{question}

\section{Acknowledgment}
We are grateful to the anonymous reviewer for their helpful comments, which increased the quality of this paper.

\bibliographystyle{plain}
\bibliography{DFA}

\end{document}